\documentclass[aps,onecolumn,a4paper,superscriptaddress,longbibliography,nofootinbib,notitlepage,groupedaddress]{revtex4-1}

\usepackage{lipsum}
\usepackage{setspace}

\usepackage{graphicx}
\usepackage{amsmath}
\usepackage{amssymb}
\usepackage{amsthm}
\usepackage[utf8]{inputenc}
\usepackage{comment}

\usepackage{color}
\usepackage[dvipsnames]{xcolor}

\newcommand{\norm}[1]{\left\lVert#1\right\rVert} 

\usepackage[a4paper,total={6.5in,9in}]{geometry}
\usepackage[english]{babel}
\usepackage[T1]{fontenc}
\usepackage{mathrsfs}
\usepackage{braket}
\usepackage{physics}
\usepackage{enumerate}
\usepackage{graphicx}
\usepackage{mathtools}
\usepackage[export]{adjustbox}
\usepackage{microtype}

\theoremstyle{plain}

\usepackage[colorlinks = true,
            linkcolor = red,
            urlcolor  = blue,
            citecolor = blue,
            anchorcolor = red]{hyperref}

\newtheorem{theorem}{Theorem}
\newtheorem{corollary}[theorem]{Corollary}
\newtheorem{proposition}[theorem]{Proposition}
\newtheorem{lemma}[theorem]{Lemma}
 
\newtheorem{definition}[theorem]{Definition}  
\newtheorem{remark}[theorem]{Remark} 
\newtheorem{problem}[theorem]{problem} 
\newtheorem*{remark*}{Remark}   
\renewcommand\qedsymbol{$\blacksquare$}
\newenvironment{proof-of}[1][{\hspace{-\blank}}]{{\medskip\noindent\textit{Proof~{#1}.\ }}}{\hfill\qedsymbol}

\renewcommand{\Tr}{{\operatorname{Tr}\,}}

\renewcommand{\dim}{{\operatorname{dim}}}
\newcommand{\id}{{\operatorname{id}}}
\newcommand{\proj}[1]{|#1\rangle\!\langle #1|}

\newcommand{\nc}{\newcommand}
\nc{\rnc}{\renewcommand}

\nc{\avg}[1]{\langle#1\rangle}
\nc{\Rank}{\operatorname{Rank}}
\nc{\smfrac}[2]{\mbox{$\frac{#1}{#2}$}}

\nc{\ox}{\otimes}
\nc{\dg}{\dagger}
\nc{\da}{\downarrow}
\nc{\ua}{\uparrow}
\nc{\cA}{{\cal A}}
\nc{\cB}{{\cal B}}
\nc{\cC}{{\cal C}}
\nc{\cF}{{\cal F}}
\nc{\cG}{{\cal G}}
\nc{\cH}{{\cal H}}
\nc{\cI}{{\cal I}}
\nc{\cJ}{{\cal J}}
\nc{\cK}{{\cal K}}
\nc{\cL}{{\cal L}}
\nc{\cM}{{\cal M}}
\nc{\cN}{{\cal N}}
\nc{\cO}{{\cal O}}
\nc{\cP}{{\cal P}}
\nc{\cQ}{{\cal Q}}
\nc{\cR}{{\cal R}}
\nc{\cS}{{\cal S}}
\nc{\cX}{{\cal X}}
\nc{\cY}{{\cal Y}}
\nc{\cW}{{\cal W}}
\nc{\cZ}{{\cal Z}}
\nc{\SWAP}{{\mathrm{SWAP}}}
\nc{\csupp}{{\operatorname{csupp}}}
\nc{\qsupp}{{\operatorname{qsupp}}}

\begin{document}

\title{Additivity Results for the R\'enyi-2 Entanglement of Purification}

\author{Shokoufe Faraji$^{1,2}$}
\email{s3faraji@uwaterloo.ca}
\author{Zahra Baghali Khanian$^{2,3}$}
 \email{zbkhanian@gmail.com}
\affiliation{$^{1}$Department of Physics and Astronomy, University of Waterloo, ON, Canada}
\affiliation{$^{2}$Perimeter Institute for Theoretical Physics, ON, Canada, N2L 2Y5}

\affiliation{$^{3}$Institute for Quantum Computing, University of Waterloo, ON, Canada, N2L 3G1}

\begin{abstract}

We reformulate the R\'enyi entanglement of purification as a constrained
minimum output R\'enyi entropy problem. Equivalently, for $p>1$, this
formulation can be expressed in terms of a constrained maximal output Schatten
$p$-norm. More precisely, for a completely positive map
$\Omega:\cL(B')\to\cL(A)$, we consider the quantity $\upsilon_p(\Omega)$
defined by optimizing $||(\Omega\otimes \id_E)(\sigma^{B'E})||_p$ over all
bipartite states $\sigma^{B'E}$ whose $B'$-marginal is maximally mixed.
We focus on the case $p=2$. First, we compute $\upsilon_2$ for the transpose depolarizing channel and prove that multiplicative under tensor powers. 
We then establish a general multiplicativity criterion: whenever a
completely positive map $\cN:\cL(B')\to\cL(A)$ satisfies
$\cN^\dagger\circ\cN=a\,\id_A+b\,\Tr[\cdot]\,I_d$ for some constants
$a,b\ge 0$, where $\cN^\dagger$ denotes the Hilbert-Schmidt adjoint of
$\cN$, the quantity $\upsilon_2(\cN)$ is multiplicative under tensor powers. 
Examples of channels satisfying this criterion include the transpose-depolarizing
channel, the depolarizing channel, and their respective complementary channels.
Further, we show that for every completely positive map $\Omega$, multiplicativity of $\upsilon_p(\Omega)$ implies multiplicativity for its  complementary map.
This yields the corresponding additivity statements for the
associated R\'enyi-2 entanglement of purification.

%

 
\end{abstract}

\maketitle


\section{Introduction and Problem Statement}


The entanglement of purification, introduced by Terhal, Horodecki, Leung, and DiVincenzo, provides a measure of the total correlations contained in a bipartite mixed state, including both classical and quantum contributions \cite{Terhal_2002}. Let $\rho^{AC}$ be a bipartite state and let
$\ket{\psi}^{ACB}$ be a purification of $\rho^{AC}$. 
Since the purifying system $B$ is not fixed uniquely, one may optimize over all possible ways of decomposing it into two subsystems via an isometry $U: B\to  E \otimes E'$. The entanglement of purification is then defined as the minimum entropy across the bipartition $AE:CE'$, namely
\begin{equation}
 \mathrm{EoP}(A:C)_\rho
:=
\inf_{U \, \textrm{isometry}} S(AE)_{\varphi}.   
\end{equation}
where the entropy is with respect to the state  $\ket{\varphi}^{ACEE'}=(I_{AC} \otimes U)\ket{\psi}^{ACB}$.

Equivalently, this optimization can be formulated in terms of quantum channels. 
Every decomposition of the purifying system, together with an isometry $U: B\to E \otimes E'$, induces a completely positive trace-preserving (CPTP) map $\Lambda:\cL(B)\to \cL(E)$ after tracing out $E'$. Conversely, by Stinespring dilation, every such channel arises in this way. Hence
\begin{equation}
\mathrm{EoP}(A:C)_\rho
=
\inf_{\Lambda  \, \mathrm{CPTP}}
S\,\left((\id_A\otimes \Lambda)(\rho^{AB})\right),    
\end{equation}
where $\rho^{AB}:=\Tr_C\ketbra{\psi}^{ACB}$. This channel formulation is particularly convenient for introducing R\'{e}nyi generalizations: for $p>0$, $p\neq 1$, one replaces the von Neumann entropy by the R\'{e}nyi entropy $S_p$, leading to
\begin{equation}\label{eq: EoP_p}
\mathrm{EoP}_{p}(A:C)_\rho
:=
\inf_{\Lambda  \, \mathrm{CPTP}}
S_p\,\left((\id_A\otimes \Lambda)(\rho^{AB})\right).
\end{equation}
We reformulate the R\'enyi generalization of the entanglement of purification as follows.
For any finite dimension bipartite quantum state $\rho^{AB}$, there exists a
completely positive (CP) map $\Omega:\mathcal L(B')\to \mathcal L(A)$ such that
\begin{equation}
\rho^{AB}=(\Omega\otimes \id_{B})(\Phi^{B'B}),
\end{equation}
where
\begin{equation}
\ket{\Phi}^{B'B}=\frac{1}{\sqrt{d_B}}
\sum_{i=1}^{d_B}\ket{i}^{B'}\ket{i}^{B},
\qquad
\Phi^{B'B}=\proj{\Phi}^{B'B},
\end{equation}
and $d_B=\dim B$. The map $\Omega$ is not necessarily trace-preserving.
A proof of this statement is provided in Lemma~\ref{lemma: state-CP map} in the appendix.
We use this representation of a quantum state to obtain
 \begin{align}
\inf_{\Lambda:\mathcal L(B)\to \mathcal L(E)\ \mathrm{CPTP}}
S_p\!\left((\id_A\otimes\Lambda)(\rho^{AB})\right)
&=\inf_{\Lambda:\mathcal L(B)\to \mathcal L(E)\ \mathrm{CPTP}}
S_p\!\left((\Omega\otimes\id_E)(\id_{B'}\otimes\Lambda)(\Phi^{B'B})\right).
\end{align}   
where the equality follows from the fact that the above superoperators commute.
This is reminiscent of the minimum output R\'enyi entropy of the completely
positive map $(\Omega \otimes \id_E)$ \cite{WernerHolevo2002,AmosovHolevoWerner2000,Shor2004,
HaydenWinter2008,Aubrun_Werner_Renyi2010}. However, in the present setting the
optimization is not over all input states. Rather, the input
states are restricted to those of the form
\begin{equation}
\sigma^{B'E}=(\id_{B'}\otimes \Lambda)(\Phi^{B'B}),
\end{equation}
for a CPTP map $\Lambda$. Thus, the quantity can be viewed as a
constrained version of the usual minimum output R\'enyi entropy.
Equivalently, for $p>1$, this optimization can be reexpressed as a constrained maximal output $p$-purity of Schatten $p$-norm for $\Omega\otimes \id_B$ $(\Omega \otimes \id_E)$:
\begin{align}
\inf_{\Lambda \, \mathrm{CPTP}}
S_{p}\,\left((\id_A \otimes \Lambda)(\rho^{AB})\right)
&=
\frac{p}{1-p}
\log \;
\sup_{\Lambda \, \mathrm{CPTP}} \;
\left\|(\Omega \otimes \id_E)
(\id_{B'} \otimes \Lambda)
(\Phi^{B'B})
\right\|_{p}, 
\end{align}
where, for $p>1$, $\norm{X}_p := \left(\Tr |X|^p\right)^{1/p} = \left(\Tr \left(X^\dagger X\right)^{p/2}\right)^{1/p}$ denotes the Schatten $p$-norm. In fact, as one of the central questions in quantum information theory is how strongly  a CP map degrades the purity of input states. A standard measure of output purity is the maximal output Schatten $p$-norm
\begin{equation}
\nu_p(\Phi):=\sup_{\rho}\|\Phi(\rho)\|_p,\qquad 1\le p\le \infty,
\end{equation}
where, by convexity, the supremum may be restricted to pure input states. The corresponding multiplicativity problem asks whether
\begin{equation}
\nu_p(\Omega\otimes\Omega')=\nu_p(\Omega)\nu_p(\Omega')
\end{equation}
holds for pairs of CP maps $\Omega$ and $\Omega'$. This problem was formulated by Amosov, Holevo, and Werner for completely positive trace-preserving maps (CPTP), and it is closely connected to the additivity problem for minimum output entropy and to other foundational additivity questions in quantum information theory \cite{AmosovHolevoWerner2000,Holevo2006AdditivityProblem,King2002UnitalQubitChannels,KingRuskai2004MaximalPNormsP2,Shor2004,MatsumotoShimonoWinter2004Additivity,FannesEtAl2004,AlickiFannes2004WernerHolevoRenyi,Watrous2005SuperOperatorNorms,DattaRuskai2005AsymmetricUnitalQudit,HaydenLeungWinter_GenericEntanglement2006,Datta2006IsotropicSpinChannels,Hayden2007MaximalPNormFalse,Ruskai2007OpenProblems,Hastings2009Superadditivity,BrandaoHorodecki2010HastingsCounterexamples,FukudaKingMoser2010HastingsComments,AubrunWerner_Hastings2011}.

Although multiplicativity has been established for several important classes of channels, it is now known to fail in full generality. Werner and Holevo exhibited a counterexample for sufficiently large values of $p$, and Hayden and Winter later proved that counterexamples exist for every $p>1$ \cite{WernerHolevo2002,HaydenWinter2008,Aubrun_Werner_Renyi2010}. These negative results make it especially natural to identify structured families of channels for which multiplicativity does survive and to understand the mechanisms behind such behavior.

In this paper, we study multiplicativity properties of the following quantity for a CP map $\Omega:\mathcal L(B')\to \mathcal L(A)$:
\begin{align}
\upsilon_p(\Omega):=\sup_{\Lambda \, \mathrm{CPTP}} \;
\left\|
(\Omega \otimes \id_E)
(\id_{B'} \otimes \Lambda)
(\Phi^{B'B})
\right\|_{p}.
\end{align}
Equivalently, using the Choi-Jamio{\l}kowski isomorphism, this quantity can be expressed as
\begin{align}\label{eq: upsilon_p Choi}
\upsilon_p(\Omega):=\sup_{\sigma^{B'E}: \,
\Tr_E(\sigma^{B'E})=I_{B'}/d_{B'}} \;
\left\|
(\Omega \otimes \id_E)\sigma^{B'E}
\right\|_{p}.
\end{align}
The multiplicativity property translates directly into additivity of the R\'enyi entropy in Eq.~\eqref{eq: EoP_p}.
A key distinction from the standard maximal $p$-norm problem is that the input to $\Omega\otimes \id_B$ cannot, in general, be restricted to pure states; rather, it is constrained by the structure of the underlying optimization problem. Nevertheless, the optimization can be restricted to the extreme points of the set of CPTP maps. Indeed, the objective function is convex in $\Lambda$, hence the optimization is taken as a supremum over all finite dimensional output systems $E$ and all CPTP maps $\Lambda:\mathcal L(B)\to\mathcal L(E)$.



An important example studied in the context of the maximal $p$-norm is the transpose-depolarizing channel
\begin{equation}
\Gamma_t(X)=tX^T+(1-t)\operatorname{Tr}(X)\frac{I_d}{d}. \qquad
-\frac{1}{d-1}\le t\le \frac{1}{d+1}.
\end{equation}
This family contains the Werner-Holevo channel as the endpoint $t=-1/(d-1)$ and has already played a significant role in the study of multiplicativities.
For this Werner-Holevo endpoint, Datta proved multiplicativity for all $1 < p\leq 2$ \cite{Datta2004WernerHolevo}.
For the additivity of the minimum output entropy, Fannes \textit{et al.} established the result for a substantial part of the parameter range, and Datta, Holevo, and Suhov later completed the proof for the full completely positive range \cite{FannesEtAl2004,DattaHolevoSuhov2006}.

\smallskip


Our main results are as follows. First, for the transpose-depolarizing channel $\Gamma_t$, we compute $\upsilon_2(\Gamma_t)$ exactly and prove that it is multiplicative under tensor powers , namely
\begin{equation}
\upsilon_2(\Gamma_t^{\otimes n})=\bigl(\upsilon_2(\Gamma_t)\bigr)^n,
\qquad \forall n\in\mathbb N.
\end{equation}
It means joint processing across several copies does not improve the optimal quantity beyond product strategies.

Second, for any completely positive map $\mathcal N$ and its Hilbert-Schmidt adjoint $\mathcal{N}^\dagger$ satisfying
\begin{equation}\label{eq:eq:general-N-assumption}
\mathcal N^\dagger\circ\mathcal N=a\,\id_A+b\,\Tr[\cdot]\,I_d,
\qquad a,b\ge 0,
\end{equation}
we compute $\upsilon_2(\mathcal N)$ and prove that
\begin{equation}
\upsilon_2(\mathcal N^{\otimes n})=\bigl(\upsilon_2(\mathcal N)\bigr)^n,
\qquad \forall n\in\mathbb N.
\end{equation}
As further examples, we show that the complementary channel
\(\Gamma_t^c(\cdot)\) of the transpose-depolarizing channel, the depolarizing
channel \(\Delta_p(\cdot)\), and its complementary channel
\(\Delta_p^c(\cdot)\) all satisfy this general criterion. Consequently,
\(\upsilon_2\) is multiplicative for each of these channels.


Third, we prove that for every completely positive map $\Omega$, multiplicativity of $\upsilon_p(\Omega)$ implies multiplicativity for the corresponding complementary map.

\smallskip

The paper is organized as follows. Section \ref{sec:norm} derives the exact formula for $\upsilon_2(\Gamma_t)$ by establishing matching lower and upper bounds. Section \ref{sec:multiplicity} proves that $\upsilon_2(\Gamma_t)$ is multiplicative under tensor powers. Section \ref{sec:general-criterion} establishes a general $p=2$ multiplicativity criterion for completely positive maps $N$ satisfying Eq.~\eqref{eq:eq:general-N-assumption}, proves that $\upsilon_p$ is invariant under passage to complementary completely positive maps, and applies these results to the transpose-depolarizing channel, the depolarizing channel, and their complementary channels. Section \ref{sec:summary} summarizes the main results and discusses consequences and open problems. 

\smallskip

In each section we explicitly specify the notation for states and maps.

\section{Norm of $\upsilon_2(\Gamma_t)$} \label{sec:norm}
We derive the explicit lower and upper bounds, their equality yields the exact norm. We start off with stating the question precisely as follows.


For $d\ge 2$, let $A\simeq \mathbb{C}^d$, $A'\simeq \mathbb{C}^d$, and let $B$ be an arbitrary finite dimensional output system. Fix an orthonormal basis $\{\lvert i\rangle\}_{i=1}^{d}$ of $A$ and the corresponding bases of $A'$. Define the normalized maximally entangled vector and state
\begin{equation}
\lvert\Phi_d\rangle=\frac{1}{\sqrt{d}}\sum_{i=1}^{d}\lvert i\rangle_A\otimes \lvert i\rangle_{A'},
\qquad
\Phi_d=\lvert\Phi_d\rangle\langle \Phi_d\rvert\in \mathcal{L}(A\otimes A').
\end{equation}
Let $T$ denote transpose with respect to this fixed basis. For a parameter $t$ in the CP-range
\begin{equation}
-\frac{1}{d-1}\le t\le \frac{1}{d+1},
\end{equation}
define the CPTP map as $\Gamma_t:\mathcal{L}(A)\to \mathcal{L}(A)$ by
\begin{equation}
\Gamma_t(X)=t\,X^{T}+(1-t)\,\mathrm{Tr}(X)\,\frac{I_d}{d}.
\end{equation}
(Note that the Werner-Holevo channel is the endpoint $t=-\frac{1}{d-1}$.)
Let $B$ be arbitrary finite dimension, and let $\Lambda:\mathcal{L}(A')\to \mathcal{L}(B)$ range over CPTP maps. Define
\begin{equation}
\upsilon_p(\Gamma_t):=
\sup_{\Lambda:\mathcal{L}(A')\to \mathcal{L}(B)\ \mathrm{CPTP}}
\left\|
(\Gamma_t\otimes \id_B)\bigl((\id_A\otimes \Lambda)(\Phi_d)\bigr)
\right\|_p .
\end{equation}
Equivalent Choi-Jamio\l kowski reformulation: Since $(\id_A\otimes \Lambda)(\Phi_d)$ is the normalized Choi state of $\Lambda$, the optimization is equivalently
\begin{equation}
\upsilon_p(\Gamma_t)=
\sup_{\sigma_{AB}\ge 0}\ 
\bigl\|(\Gamma_t\otimes \id_B)(\sigma_{AB})\bigr\|_p
\quad\text{s.t.}\quad
\mathrm{Tr}(\sigma_{AB})=1,\ \ \sigma_A=\frac{I_d}{d}.
\end{equation}
Conversely, for any finite dimension output system $B$ and any $\sigma_{AB}\ge 0$ with $\mathrm{Tr}\,\sigma_{AB}=1$ and $\sigma_A=I_d/d$, define
\begin{equation}
\Lambda_\sigma(X):=d\,\mathrm{Tr}_A\,\bigl[(X^{T}\otimes I_B)\,\sigma_{AB}\bigr].
\end{equation}
Then $\Lambda_\sigma$ is CPTP and satisfies $(\id\otimes \Lambda_\sigma)(\Phi_d)=\sigma_{AB}$. This is the Choi isomorphism in the normalization convention.


Before start with derivation of $\upsilon_2(\Gamma_t^{\otimes n})$, we first show that taking $\sup$ is equivalent to taking $\max$.
\begin{proposition}[Finite output reduction]\label{prop:finite-output-reduction}
Let
\begin{equation}
d_n:=d^n,
\qquad
D_n:=d_n^2=d^{2n}.
\end{equation}
Then
\begin{equation}\label{eq:finite-output-reduction}
\upsilon_2(\Gamma_t^{\otimes n})
=\max_{\Lambda^{(n)}:\mathcal{L}(A'^{(n)})\to \mathcal{L}(\mathbb{C}^{D_n})\ \mathrm{CPTP}} \left\|(\Gamma_t^{\otimes n}\otimes \id_{D_n})
\Bigl((\id_{A^{(n)}}\otimes \Lambda^{(n)})(\Phi_d^{\otimes n})\Bigr)
\right\|_2.
\end{equation}
In particular, the free output supremum is actually a maximum.
\end{proposition}
\begin{proof}
For each output dimension $m\in\mathbb{N}$, let $ \mathrm{CPTP}(d_n,m)$ denote the compact convex set of CPTP maps
\begin{equation}
\Lambda:\mathcal{L}(A'^{(n)})\to\mathcal{L}(\mathbb{C}^m).
\end{equation}
Define
\begin{equation}
f_m(\Lambda):= \left\|(\Gamma_t^{\otimes n}\otimes \id_m)
\Bigl((\id_{A^{(n)}}\otimes \Lambda)(\Phi_d^{\otimes n})\Bigr)
\right\|_2 .
\end{equation}
Since $\Lambda\mapsto (\id\otimes\Lambda)(\Phi_d^{\otimes n})$ is affine and the
Hilbert-Schmidt norm is continuous and convex, $f_m$ is continuous and convex on
$\mathrm{CPTP}(d_n,m)$.

Fix $m$ and $\Lambda\in \mathrm{CPTP}(d_n,m)$. Since $\mathrm{CPTP}(d_n,m)$ is a compact convex subset of a finite dimension vector space and $f_m$ is convex, there exists an extreme channel $\Lambda_{\mathrm{ext}}\in \mathrm{CPTP}(d_n,m)$ such that
\begin{equation}
f_m(\Lambda_{\mathrm{ext}})\ge f_m(\Lambda),
\end{equation}
and choose a Kraus representation
\begin{equation}
\Lambda_{\mathrm{ext}}(X)=\sum_{i=1}^r K_i X K_i^\dagger .
\end{equation}
Because $\Lambda_{\mathrm{ext}}$ is an extreme CPTP map, Choi's extremality criterion implies
\begin{equation}
r\le d_n.
\end{equation}
Let
\begin{equation}
S:=\sum_{i=1}^r \operatorname{ran}(K_i)\subseteq \mathbb{C}^m.
\end{equation}
Then
\begin{equation}
\dim \,S\le r\,d_n\le d_n^2=D_n.
\end{equation}
Let $W:\mathbb{C}^s\to\mathbb{C}^m$ be an isometry onto $S$, where $s:=\dim S$, and write
\begin{equation}
K_i = W L_i
\end{equation}
for suitable operators $L_i:\mathbb{C}^{d_n}\to\mathbb{C}^s$. Define
\begin{equation}
\widetilde{\Lambda}(X):=\sum_{i=1}^r L_i X L_i^\dagger .
\end{equation}
Then $\widetilde{\Lambda}$ is CPTP and
\begin{equation}
\Lambda_{\mathrm{ext}}(X)=W\,\widetilde{\Lambda}(X)\,W^\dagger
\qquad
\text{for all }X.
\end{equation}
Therefore
\begin{equation}
(\id\otimes \Lambda_{\mathrm{ext}})(\Phi_d^{\otimes n})
=(I\otimes W)\,
(\id\otimes \widetilde{\Lambda})(\Phi_d^{\otimes n})
\,(I\otimes W^\dagger),
\end{equation}
and hence
\begin{equation}
(\Gamma_t^{\otimes n}\otimes \id_m)
\bigl((\id\otimes \Lambda_{\mathrm{ext}})(\Phi_d^{\otimes n})\bigr)
=(I\otimes W)\, (\Gamma_t^{\otimes n}\otimes \id_s)
\bigl((\id\otimes \widetilde{\Lambda})(\Phi_d^{\otimes n})\bigr)
\,(I\otimes W^\dagger).
\end{equation}
Since Schatten norms are invariant under isometries,
\begin{equation}
f_m(\Lambda_{\mathrm{ext}})=f_s(\widetilde{\Lambda}),
\qquad s\le D_n.
\end{equation}
Finally, embed $\mathbb{C}^s$ isometrically into the fixed space $\mathbb{C}^{D_n}$.
This does not change the value of the objective. Thus every admissible value of the free-output problem is attained already by a CPTP map with output space $\mathbb{C}^{D_n}$. Consequently,

\begin{align}
\upsilon_2(\Gamma_t^{\otimes n})
&=\sup_{\substack{B^{(n)}\ \mathrm{finite\mbox{-}dimensional}\\
\Lambda^{(n)}:\mathcal{L}(A'^{(n)})\to \mathcal{L}(B^{(n)})\ \mathrm{CPTP}}}
\left\|
(\Gamma_t^{\otimes n}\otimes \id_{B^{(n)}})
\Bigl((\id_{A^{(n)}}\otimes \Lambda^{(n)})(\Phi_d^{\otimes n})\Bigr)
\right\|_2\nonumber\\
&=\max_{\Lambda^{(n)}:\mathcal{L}(A'^{(n)})\to \mathcal{L}(\mathbb{C}^{D_n})\ \mathrm{CPTP}}
\left\|
(\Gamma_t^{\otimes n}\otimes \id_{D_n})
\Bigl((\id_{A^{(n)}}\otimes \Lambda^{(n)})(\Phi_d^{\otimes n})\Bigr)
\right\|_2.
\end{align}
which is Eq.~\eqref{eq:finite-output-reduction}, and since $\mathrm{CPTP}(d_n,D_n)$ is compact and $f_{D_n}$ is continuous, the supremum over maps into $\mathbb{C}^{D_n}$ is attained.
\end{proof}

\subsection{Stating the norm in terms of Choi states}

Let $B$ be a finite dimension output system, and let $\Lambda:\mathcal{L}(A')\to \mathcal{L}(B)$ be a CPTP map. Define its normalized Choi state (relative to $\Phi_d$) by
\begin{equation}\label{eq:def-choi-sigma}
\sigma_{AB} := (\id_A\otimes \Lambda)(\Phi_d)\ \in\ \mathcal{L}(A\otimes B).
\end{equation}
We also use the standard marginal notation $\sigma_A:=\mathrm{Tr}_B(\sigma_{AB})$ and $\sigma_B:=\mathrm{Tr}_A(\sigma_{AB})$.
\begin{proposition}\label{prop:choi-feasible-set}
Let
\begin{equation}\label{eq:def-feasible-F}
\mathcal{F}_d(B) := \Bigl\{\sigma_{AB}\in \mathcal{L}(A\otimes B)\ :\ \sigma_{AB}\ge 0,\ \mathrm{Tr}(\sigma_{AB})=1,\ \sigma_A=\tfrac{I_d}{d}\Bigr\}.
\end{equation}
Then the map $\Lambda\mapsto\sigma_{AB}=(\id_A\otimes\Lambda)(\Phi_d)$ is an affine bijection between the set of CPTP maps $\Lambda:\mathcal{L}(A')\to\mathcal{L}(B)$ and $\mathcal{F}_d$. Consequently,
\begin{equation}\label{eq:nu-choi-form}
\upsilon_p(\Gamma_t)
=\max_{\sigma_{AB}\in \mathcal{F}_d}\ \bigl\|(\Gamma_t\otimes \id_B)(\sigma_{AB})\bigr\|_p.
\end{equation}
\end{proposition}
\begin{proof}
(i) Every CPTP map produces a point of $\mathcal{F}_d$:
Let $\Lambda$ be CPTP and define $\sigma_{AB}$ as in relation Eq.~\eqref{eq:def-choi-sigma}.
Since $\Phi_d\ge 0$ and $\id_A\otimes \Lambda$ is completely positive, therefore $\sigma_{AB}\ge 0$.
Moreover, because $\id_A\otimes\Lambda$ is trace-preserving we have $\mathrm{Tr}(\sigma_{AB})=\mathrm{Tr}(\Phi_d)=1$. Now, by using trace-preservation of $\Lambda$, we obtain
\begin{equation}
\sigma_A =\mathrm{Tr}_B\bigl((\id_A\otimes \Lambda)(\Phi_d)\bigr) =(\id_A\otimes \mathrm{Tr})(\Phi_d)=\frac{I_d}{d},
\end{equation}
where the last identity is a direct computation from the definition of $\Phi_d$. Thus, $\sigma_{AB}\in \mathcal{F}_d(B)$.

\smallskip
(ii) Every point of $\mathcal{F}_d$ is the Choi state of a unique CPTP map:
Fix a finite dimension output system $B$ and $\sigma_{AB}\in \mathcal{F}_d$ and define a linear map $\Lambda_\sigma:\mathcal{L}(A')\to\mathcal{L}(B)$ by
\begin{equation}\label{eq:def-Lambda-sigma}
\Lambda_\sigma(X):= d\,\mathrm{Tr}_A\Bigl[(X^{T}\otimes I_B)\,\sigma_{AB}\Bigr],\qquad X\in \mathcal{L}(A').
\end{equation}
By construction, $\Lambda_\sigma$ is completely positive (this is the Choi theorem: $\sigma_{AB}\ge 0$ is exactly the Choi matrix of $\Lambda_\sigma$ under the normalization fixed by $\Phi_d$). Also, $\Lambda_\sigma$ is trace-preserving, since
\begin{equation}
\mathrm{Tr}\bigl(\Lambda_\sigma(X)\bigr)
=d\,\mathrm{Tr}\Bigl[(X^{T}\otimes I_B)\,\sigma_{AB}\Bigr]
=d\,\mathrm{Tr}\Bigl[X^{T}\,\sigma_A\Bigr]
=d\,\mathrm{Tr}\Bigl[X^{T}\,\tfrac{I_d}{d}\Bigr]
=\mathrm{Tr}(X).
\end{equation}
Thus, $\Lambda_\sigma$ is CPTP.

Now, we need to show that $\sigma_{AB}=(\id_A\otimes \Lambda_\sigma)(\Phi_d)$. By using $\Phi_d=\frac{1}{d}\sum_{i,j=1}^d |i\rangle\langle j|\otimes |i\rangle\langle j|$ and relation Eq.~\eqref{eq:def-Lambda-sigma}, we have
\begin{equation}
(\id_A\otimes \Lambda_\sigma)(\Phi_d)
=\frac{1}{d}\sum_{i,j=1}^d |i\rangle\langle j|\otimes \Lambda_\sigma(|i\rangle\langle j|)
=\sum_{i,j=1}^d |i\rangle\langle j|\otimes \mathrm{Tr}_A\Bigl[(|j\rangle\langle i|\otimes I_B)\sigma_{AB}\Bigr].
\end{equation}
Writing $\sigma_{AB}=\sum_{k,\ell=1}^d |k\rangle\langle \ell|\otimes \sigma_{k\ell}$ in $A$-blocks, the partial trace identity
$\mathrm{Tr}_A[(|j\rangle\langle i|\otimes I_B)\sigma_{AB}]=\sigma_{ij}$
yields
\begin{equation}
(\id_A\otimes \Lambda_\sigma)(\Phi_d)
=\sum_{i,j=1}^d |i\rangle\langle j|\otimes \sigma_{ij}=\sigma_{AB}.
\end{equation}
Finally, uniqueness follows because the Choi correspondence is injective under the fixed normalization: if $(\id\otimes\Lambda_1)(\Phi_d)=(\id\otimes\Lambda_2)(\Phi_d)$ then $\Lambda_1=\Lambda_2$ by the reconstruction formula Eq.~\eqref{eq:def-Lambda-sigma}.
This proves the affine bijection This proves the affine bijection for each fixed finite dimension output system $B$, and hence the reformulation \eqref{eq:nu-choi-form}.
\end{proof}

\begin{remark}\label{rem:nu2-attainment}
For each fixed finite dimension output system $B$, the set $\mathcal{F}_d(B)$ is compact and the map
\begin{equation}
\sigma\mapsto \|(\Gamma_t\otimes \id_B)(\sigma)\|_2
\end{equation}
is continuous. Hence, for each fixed $B$, the inner maximum in \eqref{eq:nu-choi-form} is attained.
\end{remark}
Now for simplicity, we express the output $(\Gamma_t\otimes \id_B)(\sigma_{AB})$ directly in terms of $\sigma_{AB}$ and its marginals as \begin{equation}\label{eq:affine-output}
(\Gamma_t\otimes \id_B)(\sigma_{AB})
=t\,\sigma_{AB}^{T_A}+(1-t)\,\frac{I_d}{d}\otimes \sigma_B,
\end{equation}
where $T_A$ denotes transpose on subsystem $A$ with respect to the fixed basis $\{|i\rangle\}_{i=1}^d$.


\subsection{Commuting the local maps and introducing a fixed seed operator}

The definition of $\omega_{AB}(\Lambda)$ involves first applying $\id_A\otimes \Lambda$ to the maximally entangled state and then applying $\Gamma_t\otimes \id_B$. Since $\Gamma_t$ acts only on the $A$-register while $\Lambda$ acts only on the $A'$-register, these operations commute in the natural sense. This allows us to rewrite the optimization as a post-processing problem: apply an arbitrary channel to one subsystem of a fixed bipartite operator determined solely by $(\Gamma_t)$. Let's define the seed operator as
\begin{equation}\label{eq:def-seed-tau}
\tau_{AA'} := (\Gamma_t\otimes \id_{A'})(\Phi_d)\ \in\ \mathcal{L}(A\otimes A').
\end{equation}
Because of the commuting operators we can rewrite $\upsilon_p(\Gamma_t)$ as
\begin{equation}\label{eq:omega-as-postprocessing}
\omega_{AB}(\Lambda) =(\id_A\otimes \Lambda)(\tau_{AA'}).
\end{equation}
Consequently,
\begin{equation}\label{eq:nu-as-postprocessing}
\upsilon_p(\Gamma_t)=\max_{\Lambda\ \mathrm{CPTP}}
\ \bigl\|(\id_A\otimes \Lambda)(\tau_{AA'})\bigr\|_p.
\end{equation}
Now problem reduces to maximizing the output Schatten norm of $(\id_A\otimes \Lambda)(\tau_{AA'})$. To proceed, we require an explicit expression for $\tau_{AA'}$ in the fixed basis. 
This will later allow us to analyze its spectrum and identify which structural features of  $\tau_{AA'}$ can be exploited by the optimizing channel $\Lambda$. Further, the operator $\tau_{AA'}$ defined in Eq.~\eqref{eq:def-seed-tau} admits the expansion
\begin{equation}\label{eq:tau-explicit}
\tau_{AA'}=\frac{t}{d}\sum_{i,j=1}^d
\bigl(\,|j\rangle\langle i|\,\bigr)_A\ \otimes\ \bigl(\,|i\rangle\langle j|\,\bigr)_{A'}
\;+\; \frac{1-t}{d^2}\, I_d\otimes I_d.
\end{equation}


\subsubsection{Spectral decomposition of the seed operator $\tau_{AA'}$}

The reformulation Eq.~\eqref{eq:nu-as-postprocessing} shows that the optimization concerns the action of $\id_A\otimes\Lambda$ on the fixed operator $\tau_{AA'}$. In order to obtain sharp bounds (and later to identify candidate optimizers), it is essential to know the spectrum of $\tau_{AA'}$ explicitly. Since $\tau_{AA'}$ is a linear combination of the identity and the flip operator, its spectrum can be computed by decomposing $A\otimes A'$ into two orthogonal invariant subspaces.

Now define the flip operator $\Pi_{AA'}\in\mathcal{L}(A\otimes A')$ by
\begin{equation}\label{eq:def-Pi}
\Pi_{AA'} \;:=\; \sum_{i,j=1}^d
\bigl(\,|j\rangle\langle i|\,\bigr)_A\ \otimes\ \bigl(\,|i\rangle\langle j|\,\bigr)_{A'},
\end{equation}
in which for every $k,\ell\in\{1,\dots,d\}$,
\begin{equation}\label{eq:Pi-action}
\Pi_{AA'}\bigl(|k\rangle_A\otimes|\ell\rangle_{A'}\bigr)
=|\ell\rangle_A\otimes|k\rangle_{A'}.
\end{equation}
In particular, $\Pi_{AA'}$ is Hermitian and unitary, and satisfies $\Pi_{AA'}^2=I_d\otimes I_d$. Then Eq.~\eqref{eq:tau-explicit} may be written equivalently as
\begin{equation}\label{eq:tau-as-linear-combination}
\tau_{AA'} \;=\; \frac{t}{d}\,\Pi_{AA'} \;+\; \frac{1-t}{d^2}\,I_d\otimes I_d.
\end{equation}

\begin{definition}[Two invariant subspaces of $A\otimes A'$]\label{def:sym-antisym}
Define subspaces $\mathsf{S},\mathsf{A}\subseteq A\otimes A'$ by
\begin{align}
\mathsf{S}
&:= \mathrm{span}\Bigl(\ \{|i\rangle\otimes|i\rangle:\ 1\le i\le d\}\ \cup\
\Bigl\{\frac{|i\rangle\otimes|j\rangle+|j\rangle\otimes|i\rangle}{\sqrt{2}}:\ 1\le i<j\le d\Bigr\}\ \Bigr),
\label{eq:def-S}\\[2mm]
\mathsf{A}
&:= \mathrm{span}\Bigl(\ \Bigl\{\frac{|i\rangle\otimes|j\rangle-|j\rangle\otimes|i\rangle}{\sqrt{2}}:\ 1\le i<j\le d\Bigr\}\ \Bigr).
\label{eq:def-A}
\end{align}
\end{definition}

\begin{lemma}[Orthogonal decomposition and dimensions]\label{lem:orth-decomp-dim}
The subspaces $\mathsf{S}$ and $\mathsf{A}$ are orthogonal, invariant under $\Pi_{AA'}$, and
\begin{equation}
A\otimes A'=\mathsf{S}\oplus \mathsf{A}.
\end{equation}
Moreover,
\begin{equation}
\dim(\mathsf{S})=\frac{d(d+1)}{2},
\qquad
\dim(\mathsf{A})=\frac{d(d-1)}{2}.
\end{equation}
\end{lemma}
\begin{proof}
The spanning sets in \eqref{eq:def-S} and \eqref{eq:def-A} are orthonormal by direct inner product computation, and every vector from the spanning set of $\mathsf{S}$ is orthogonal to every vector from the spanning set of
$\mathsf{A}$; hence $\mathsf{S}\perp \mathsf{A}$. Counting the displayed orthonormal spanning vectors gives the stated dimensions: there are $d$ diagonal vectors $|i\rangle\otimes|i\rangle$ and $\binom{d}{2}$ symmetrized off-diagonal vectors in \eqref{eq:def-S}, gives $\dim(\mathsf{S})=\frac{d(d+1)}{2}$, and for the antisymmetrized off-diagonal $\dim(\mathsf{A})=\frac{d(d-1)}{2}.$ Therefore $\dim(\mathsf{S})+\dim(\mathsf{A})=d^2=\dim(A\otimes A'),$ and since $\mathsf{S}\perp \mathsf{A}$, we obtain
\begin{equation}
A\otimes A'=\mathsf{S}\oplus \mathsf{A}.
\end{equation}
To see invariance under $\Pi_{AA'}$, apply \eqref{eq:Pi-action} to the spanning vectors: for every $i$, $\Pi_{AA'}(|i\rangle\otimes|i\rangle)=|i\rangle\otimes|i\rangle,$ and for $1\le i<j\le d$,
\begin{equation}
\Pi_{AA'}\frac{|i\rangle\otimes|j\rangle\pm|j\rangle\otimes|i\rangle}{\sqrt2}
=\frac{|j\rangle\otimes|i\rangle\pm|i\rangle\otimes|j\rangle}{\sqrt2},
\end{equation}
Hence $\Pi_{AA'}\mathsf{S}\subseteq\mathsf{S}$ and $\Pi_{AA'}\mathsf{A}\subseteq\mathsf{A}$.
\end{proof}
\begin{lemma}[Spectrum of $\Pi_{AA'}$]\label{lem:Pi-spectrum}
The operator $\Pi_{AA'}$ has eigenvalue $+1$ on $\mathsf{S}$ and eigenvalue $-1$ on $\mathsf{A}$. Equivalently,
\begin{equation}
\Pi_{AA'}|_{\mathsf{S}} = +I_{\mathsf{S}},
\qquad \Pi_{AA'}|_{\mathsf{A}} = -I_{\mathsf{A}}.
\end{equation}
\end{lemma}

\begin{proof}
By the invariance statements in Lemma~\ref{lem:orth-decomp-dim}, it suffices to evaluate $\Pi_{AA'}$ on the spanning vectors in Eq.~\eqref{eq:def-S} and Eq.~\eqref{eq:def-A}$.$
By the last two displays in the proof of Lemma~\ref{lem:orth-decomp-dim}, $\Pi_{AA'}$ fixes each spanning vector of $\mathsf{S}$ and negates each spanning vector of $\mathsf{A}$. Since these spanning sets are bases of $\mathsf{S}$ and $\mathsf{A}$, the claimed eigenvalue relations follow.
\end{proof}
\begin{proposition}[Spectral decomposition of $\tau_{AA'}$]\label{prop:tau-spectrum}
Let $\tau_{AA'}$ be given by Eq.~\eqref{eq:tau-as-linear-combination}. Then $\tau_{AA'}$ has exactly two eigenvalues:
\begin{align}\label{eq:tau-eigs}
\lambda_{+}(\Gamma_t)&=\frac{1+t(d-1)}{d^2}
\quad \text{with multiplicity}\quad \frac{d(d+1)}{2},\nonumber\\
\lambda_{-}(\Gamma_t)&=\frac{1-t(d+1)}{d^2}
\quad\text{with multiplicity}\quad \frac{d(d-1)}{2}.
\end{align}
Moreover, $\tau_{AA'}\ge 0$ if and only if $-\frac{1}{d-1}\le t\le \frac{1}{d+1}$.
\end{proposition}

\begin{proof}
By Eq.~\eqref{eq:tau-as-linear-combination} and Lemma~\ref{lem:Pi-spectrum}, $\tau_{AA'}$ acts as a scalar on each of the invariant subspaces $\mathsf{S}$ and $\mathsf{A}$:
on $\mathsf{S}$,
\begin{equation}
\tau_{AA'}|_{\mathsf{S}}
=
\Bigl(\frac{t}{d}\cdot (+1) + \frac{1-t}{d^2}\Bigr) I_{\mathsf{S}}
=
\frac{td + (1-t)}{d^2}\,I_{\mathsf{S}}
=
\frac{1+t(d-1)}{d^2}\,I_{\mathsf{S}},
\end{equation}
and on $\mathsf{A}$,
\begin{equation}
\tau_{AA'}|_{\mathsf{A}}
=
\Bigl(\frac{t}{d}\cdot (-1) + \frac{1-t}{d^2}\Bigr) I_{\mathsf{A}}
=
\frac{-td + (1-t)}{d^2}\,I_{\mathsf{A}}
=
\frac{1-t(d+1)}{d^2}\,I_{\mathsf{A}}.
\end{equation}
Thus the eigenvalues are as in Eq.~\eqref{eq:tau-eigs}, and their multiplicities are the dimensions from Lemma~\ref{lem:orth-decomp-dim}.

Finally, $\tau_{AA'}\ge 0$ holds if and only if both eigenvalues in Eq.~\eqref{eq:tau-eigs} are nonnegative, i.e.,
$1+t(d-1)\ge 0$ and $1-t(d+1)\ge 0$, which is equivalent to
$-\frac{1}{d-1}\le t\le \frac{1}{d+1}$.
\end{proof}


\subsection{Explicit lower bounds}
For the lower bounds, we do not restrict to the case $p=2$; the same argument works for every $p\in[1,\infty]$. By Eq.~\eqref{eq:nu-as-postprocessing} we have
\begin{equation}
\upsilon_p(\Gamma_t)=\max_{\Lambda\ \mathrm{CPTP}}\ \bigl\|(\id_A\otimes\Lambda)(\tau_{AA'})\bigr\|_p,
\end{equation}
so any explicit choice of CPTP map $\Lambda$ produces an explicit lower bound on $\upsilon_p(\Gamma_t)$. We explore two benchmark choices: the identity channel and pure state constant output channels, and evaluate the resulting Schatten norms in closed form.

\begin{lemma}[Closed form for $\|\tau_{AA'}\|_p$]\label{cor:tau-p-norm}
For $1\le p<\infty$,
\begin{equation}\label{eq:tau-p-norm-finite}
\|\tau_{AA'}\|_p=\Biggl(\frac{d(d+1)}{2}\,\lambda_+(d,t)^p+\frac{d(d-1)}{2}\,\lambda_-(d,t)^p\Biggr)^{\,1/p},
\end{equation}
and for $p=\infty$,
\begin{equation}\label{eq:tau-p-norm-infty}
\|\tau_{AA'}\|_\infty=\max\{\lambda_+(d,t),\lambda_-(d,t)\},
\end{equation}
where $\lambda_\pm(d,t)$ are the eigenvalues from Proposition~\ref{prop:tau-spectrum}.
\end{lemma}

\begin{proof}
By Proposition~\ref{prop:tau-spectrum}, $\tau_{AA'}$ has eigenvalue $\lambda_+(d,t)$ with multiplicity $\frac{d(d+1)}{2}$ and eigenvalue $\lambda_-(d,t)$ with multiplicity $\frac{d(d-1)}{2}$. Since $t$ is in the CP-range, $\tau_{AA'}\ge0$ and its singular values coincide with its eigenvalues. The formula Eq.~\eqref{eq:tau-p-norm-finite} follows from the definition of the Schatten $p$-norm, and Eq.~\eqref{eq:tau-p-norm-infty} follows from the definition of the operator norm as the largest eigenvalue for positive semidefinite operators.
\end{proof}

\begin{lemma}[Pure state constant output channel]\label{lem:constant-lower-bound}
Fix a unit vector $|\psi\rangle\in B$ and define the CPTP map $\Lambda_\psi:\mathcal{L}(A')\to\mathcal{L}(B)$ by
\begin{equation}
\Lambda_\psi(X):=\mathrm{Tr}(X)\,|\psi\rangle\langle\psi|.
\end{equation}
Then
\begin{equation}\label{eq:constant-output-tau}
(\id_A\otimes \Lambda_\psi)(\tau_{AA'})
=
\frac{I_d}{d}\otimes |\psi\rangle\langle\psi|,
\end{equation}
\end{lemma}

\begin{proof}
First observe that for any $Z_{AA'}\in\mathcal{L}(A\otimes A')$, linearity and the definition of $\Lambda_\psi$ imply
\begin{equation}\label{eq:trace-out-identity}
(\id_A\otimes \Lambda_\psi)(Z_{AA'})
=\mathrm{Tr}_{A'}(Z_{AA'})\otimes |\psi\rangle\langle\psi|.
\end{equation}
Indeed, Eq.~\eqref{eq:trace-out-identity} holds on simple tensors $X\otimes Y$ because
$(\id_A\otimes\Lambda_\psi)(X\otimes Y)=X\otimes \mathrm{Tr}(Y)|\psi\rangle\langle\psi|
=\mathrm{Tr}_{A'}(X\otimes Y)\otimes|\psi\rangle\langle\psi|$,
and hence holds in general by linearity. Applying Eq.~\eqref{eq:trace-out-identity} with $Z_{AA'}=\tau_{AA'}$ gives
\begin{equation}
(\id_A\otimes \Lambda_\psi)(\tau_{AA'})
=\mathrm{Tr}_{A'}(\tau_{AA'})\otimes |\psi\rangle\langle\psi|.
\end{equation}
By using Eq.~\eqref{eq:def-seed-tau} and the fact that $\mathrm{Tr}_{A'}(\Phi_d)=I_d/d$, we obtain
\begin{equation}
\mathrm{Tr}_{A'}(\tau_{AA'}) =\mathrm{Tr}_{A'}\bigl((\Gamma_t\otimes \id_{A'})(\Phi_d)\bigr) =\Gamma_t\bigl(\mathrm{Tr}_{A'}(\Phi_d)\bigr) =\Gamma_t(I_d/d).
\end{equation}
Since $(I_d/d)^T=I_d/d$ and $\mathrm{Tr}(I_d/d)=1$, we have $\Gamma_t(I_d/d)=I_d/d$, proving Eq.~\eqref{eq:constant-output-tau}.

\end{proof}

\begin{proposition}[Two explicit lower bounds]\label{cor:two-lower-bounds}
For every $p\in[1,\infty]$ and every $t$ in the CP-range,
\begin{equation}\label{eq:lower-bound-max}
\upsilon_p(\Gamma_t)\ \ge\ \max\Bigl\{\ \|\tau_{AA'}\|_p,\ d^{\frac{1}{p}-1}\ \Bigr\},
\end{equation}
where $\|\tau_{AA'}\|_p$ is given explicitly by Lemma~\ref{cor:tau-p-norm}.
\end{proposition}

\begin{proof}
For the first lower bound, choose the admissible output system $B=A'$ and the admissible CPTP map
\begin{equation}
\Lambda=\id_{A'}:\mathcal L(A')\to\mathcal L(A').
\end{equation}
Then
\begin{equation}
(\id_A\otimes \Lambda)(\tau_{AA'})=\tau_{AA'},
\end{equation}
hence
\begin{equation}
\upsilon_p(\Gamma_t)\ge \|\tau_{AA'}\|_p.
\end{equation}
For the second lower bound, choose any finite dimensional output system $B$ and any unit vector $|\psi\rangle\in B$, and let
\begin{equation}
\Lambda_\psi(X)=\mathrm{Tr}(X)\,|\psi\rangle\langle\psi|.
\end{equation}
By Lemma \ref{lem:constant-lower-bound},
\begin{equation}
(\id_A\otimes \Lambda_\psi)(\tau_{AA'})=
\frac{I_d}{d}\otimes |\psi\rangle\langle\psi|,
\end{equation}
so
\begin{equation}
\left\|(\id_A\otimes \Lambda_\psi)(\tau_{AA'})\right\|_p=d^{\frac1p-1}.
\end{equation}
Therefore
\begin{equation}
\upsilon_p(\Gamma_t)\ge d^{\frac1p-1}.
\end{equation}
Combining the two lower bounds gives the claim.
\end{proof}

\subsection{Explicit upper bound}

The Schatten $2$-norm admits the identity $\|X\|_2^2=\mathrm{Tr}(X^\dagger X)$, which turns the maximization defining $\upsilon_2(\Gamma_t)$ into an optimization of a quadratic functional. In this subsection, we simplify this quadratic functional to a weighted combination of two purities: the global purity $\mathrm{Tr}(\sigma_{AB}^2)$ and the marginal purity $\mathrm{Tr}(\sigma_B^2)$. A sharp inequality relating these two purities for $\sigma_{AB}\in\mathcal{F}_d$ then yields a matching upper bound and finally a closed form for $\upsilon_2(\Gamma_t)$.

\begin{lemma}[A purity identity for $\|(\Gamma_t\otimes \id_B)(\sigma_{AB})\|_2^2$]\label{lem:p2-purity-identity}
Let $B$ be a finite dimension output system and set $\sigma_{AB}\in\mathcal{F}_d(B)$. Then
\begin{equation}\label{eq:p2-identity}
\|\omega_{AB}\|_2^2=t^2\,\mathrm{Tr}(\sigma_{AB}^2)
+\frac{1-t^2}{d}\,\mathrm{Tr}(\sigma_B^2).
\end{equation}
\end{lemma}

\begin{proof}
By Eq.~\eqref{eq:affine-output} we have
\begin{equation}
\omega_{AB}=t\,\sigma_{AB}^{T_A}+(1-t)\frac{I_d}{d}\otimes \sigma_B.
\end{equation}
Since $\omega_{AB}$ is Hermitian, $\|\omega_{AB}\|_2^2=\mathrm{Tr}(\omega_{AB}^2)$, hence
\begin{align*}
\|\omega_{AB}\|_2^2 &=t^2\,\mathrm{Tr}\bigl((\sigma_{AB}^{T_A})^2\bigr)
+2t(1-t)\,\mathrm{Tr}\Bigl(\sigma_{AB}^{T_A}\Bigl(\frac{I_d}{d}\otimes \sigma_B\Bigr)\Bigr)
+(1-t)^2\,\mathrm{Tr}\Bigl(\Bigl(\frac{I_d}{d}\otimes \sigma_B\Bigr)^2\Bigr).
\end{align*}
We evaluate the three traces.
\smallskip

(i) The first term.
Partial transpose is an isometry for the Hilbert-Schmidt inner product, i.e.
$\mathrm{Tr}\bigl(X^{T_A}Y^{T_A}\bigr)=\mathrm{Tr}(XY)$ for all $X,Y$.
Taking $X=Y=\sigma_{AB}$ gives
\begin{equation}
\mathrm{Tr}\bigl((\sigma_{AB}^{T_A})^2\bigr)=\mathrm{Tr}(\sigma_{AB}^2).
\end{equation}
\;\,\,(ii) The cross term.
Using $\mathrm{Tr}(X^{T_A}Z)=\mathrm{Tr}\bigl(X\,Z^{T_A}\bigr)$ and noting that
$\bigl(\frac{I_d}{d}\otimes \sigma_B\bigr)^{T_A}=\frac{I_d}{d}\otimes \sigma_B$,
we obtain
\begin{equation}
\mathrm{Tr}\Bigl(\sigma_{AB}^{T_A}\Bigl(\frac{I_d}{d}\otimes \sigma_B\Bigr)\Bigr)
=\mathrm{Tr}\Bigl(\sigma_{AB}\Bigl(\frac{I_d}{d}\otimes \sigma_B\Bigr)\Bigr)
=\frac{1}{d}\,\mathrm{Tr}\bigl(\sigma_B^2\bigr),
\end{equation}
where the last equality uses $\mathrm{Tr}_{AB}\bigl(\sigma_{AB}(I_d\otimes \sigma_B)\bigr)
=\mathrm{Tr}_B\bigl((\mathrm{Tr}_A\sigma_{AB})\sigma_B\bigr)=\mathrm{Tr}(\sigma_B^2)$.

(iii) The last term.
By multiplicativity of the Schatten $2$-norm and $\mathrm{Tr}\bigl((I_d/d)^2\bigr)=1/d$,
\begin{equation}
\mathrm{Tr}\Bigl(\Bigl(\frac{I_d}{d}\otimes \sigma_B\Bigr)^2\Bigr)
=\mathrm{Tr}\Bigl(\Bigl(\frac{I_d}{d}\Bigr)^2\Bigr)\,\mathrm{Tr}(\sigma_B^2)
=\frac{1}{d}\,\mathrm{Tr}(\sigma_B^2).
\end{equation}
Substituting (i)-(iii) back into the expansion and simplifying the coefficient of $\mathrm{Tr}(\sigma_B^2)$,
\begin{equation}
\frac{2t(1-t)}{d}+\frac{(1-t)^2}{d}=\frac{(1-t)(1+t)}{d}
=\frac{1-t^2}{d},
\end{equation}
yields Eq.~\eqref{eq:p2-identity}.
\end{proof}
\begin{lemma}[A two party purity inequality]\label{lem:purity-inequality-bipartite}
Let $\rho_{BE}$ be a density operator on $B\otimes E$, and let $\rho_B:=\mathrm{Tr}_E(\rho_{BE})$ and $\rho_E:=\mathrm{Tr}_B(\rho_{BE})$. Then

\begin{equation}\label{eq:purity-inequality-bipartite}
\mathrm{Tr}(\rho_B^2)+\mathrm{Tr}(\rho_E^2)\ \le\ 1+\mathrm{Tr}(\rho_{BE}^2).
\end{equation}
\end{lemma}

\begin{proof}
Let $B'$ and $E'$ be copies of $B$ and $E$ with the corresponding fixed bases. Define the operators
\begin{equation}
\Pi_{BB'}:=\sum_{i,j}\bigl(|j\rangle\langle i|\bigr)_B\otimes\bigl(|i\rangle\langle j|\bigr)_{B'}, \qquad
\Pi_{EE'}:=\sum_{k,\ell}\bigl(|\ell\rangle\langle k|\bigr)_E\otimes\bigl(|k\rangle\langle \ell|\bigr)_{E'}.
\end{equation}
(These are the same index-defined flip operators as in Eq.~\eqref{eq:def-Pi}, but acting on $B\otimes B'$ and $E\otimes E'$, respectively.) Now consider $\rho_{BE}\otimes \rho_{B'E'}$ on $(B\otimes E)\otimes(B'\otimes E')$. A direct index computation shows the standard identities below
\begin{equation}\label{eq:swap-identities}
\mathrm{Tr}(\rho_B^2)=\mathrm{Tr}\bigl((\rho_{BE}\otimes\rho_{B'E'})(\Pi_{BB'}\otimes I_{EE'})\bigr), \quad
\mathrm{Tr}(\rho_E^2)=\mathrm{Tr}\bigl((\rho_{BE}\otimes\rho_{B'E'})(I_{BB'}\otimes \Pi_{EE'})\bigr),
\end{equation}
and
\begin{equation}\label{eq:swap-identity-global}
\mathrm{Tr}(\rho_{BE}^2)=\mathrm{Tr}\bigl((\rho_{BE}\otimes\rho_{B'E'})(\Pi_{BB'}\otimes \Pi_{EE'})\bigr).
\end{equation}
Now note that $(I_{BB'}-\Pi_{BB'})\ge 0$ and $(I_{EE'}-\Pi_{EE'})\ge 0$ (each is a projector), hence
\begin{equation}
(I_{BB'}-\Pi_{BB'})\otimes(I_{EE'}-\Pi_{EE'})\ \ge\ 0.
\end{equation}
Expanding this inequality yields
\begin{equation}
I_{BB'}\otimes I_{EE'} -\Pi_{BB'}\otimes I_{EE'}
-I_{BB'}\otimes \Pi_{EE'} +\Pi_{BB'}\otimes \Pi_{EE'}
\ \ge\ 0,
\end{equation}
or equivalently,
\begin{equation}
\Pi_{BB'}\otimes I_{EE'}+I_{BB'}\otimes \Pi_{EE'}
\ \le\ I_{BB'}\otimes I_{EE'}+\Pi_{BB'}\otimes \Pi_{EE'}.
\end{equation}
Taking the trace of both sides against the positive operator $\rho_{BE}\otimes\rho_{B'E'}$ gives
\begin{align}
\mathrm{Tr}\bigl((\rho_{BE}\otimes\rho_{B'E'})(\Pi_{BB'}\otimes I_{EE'})\bigr)
&+\mathrm{Tr}\bigl((\rho_{BE}\otimes\rho_{B'E'})(I_{BB'}\otimes \Pi_{EE'})\bigr)\nonumber\\[6pt]
&\le \mathrm{Tr}(\rho_{BE}\otimes\rho_{B'E'})
+ \mathrm{Tr}\bigl((\rho_{BE}\otimes\rho_{B'E'})(\Pi_{BB'}\otimes \Pi_{EE'})\bigr).
\end{align}
Using Eq.~\eqref{eq:swap-identities} - Eq.~\eqref{eq:swap-identity-global} and $\mathrm{Tr}(\rho_{BE}\otimes\rho_{B'E'})=1$ yields Eq.~\eqref{eq:purity-inequality-bipartite}.
\end{proof}
\begin{corollary}[A purity tradeoff on $\mathcal{F}_d$]\label{cor:purity-tradeoff-Fd}
For every finite dimension output system $B$, and every $\sigma_{AB}\in\mathcal{F}_d$,
\begin{equation}\label{eq:purity-tradeoff-Fd}
\mathrm{Tr}(\sigma_{AB}^2)+\mathrm{Tr}(\sigma_B^2)\ \le\ 1+\frac{1}{d}.
\end{equation}
\end{corollary}

\begin{proof}
Let $\Lambda:\mathcal{L}(A')\to\mathcal{L}(B)$ be the unique CPTP map corresponding to $\sigma_{AB}$ via Proposition~\ref{prop:choi-feasible-set}. Fix a Stinespring isometry $V:A'\to B\otimes E$ such that
$\Lambda(X)=\mathrm{Tr}_E(VXV^\dagger)$ for all $X\in\mathcal{L}(A')$.
Define the pure state
\begin{equation}
|\Psi\rangle_{ABE}:=(\id_A\otimes V)|\Phi_d\rangle\in A\otimes B\otimes E, \qquad
\rho_{ABE}:=|\Psi\rangle\langle\Psi|.
\end{equation}
Then $\mathrm{Tr}_E(\rho_{ABE})=(\id_A\otimes\Lambda)(\Phi_d)=\sigma_{AB}$, therefore, $\rho_B=\mathrm{Tr}_{AE}(\rho_{ABE})=\sigma_B$.
Moreover, since $\rho_{ABE}$ is pure, the reduced states $\rho_E:=\mathrm{Tr}_{AB}(\rho_{ABE})$ and $\sigma_{AB}$ have the same nonzero spectrum, so $\mathrm{Tr}(\rho_E^2)=\mathrm{Tr}(\sigma_{AB}^2)$. Applying Lemma~\ref{lem:purity-inequality-bipartite} to $\rho_{BE}:=\mathrm{Tr}_A(\rho_{ABE})$ yields
\begin{equation}
\mathrm{Tr}(\sigma_B^2)+\mathrm{Tr}(\sigma_{AB}^2)
=\mathrm{Tr}(\rho_B^2)+\mathrm{Tr}(\rho_E^2) \le 1+\mathrm{Tr}(\rho_{BE}^2).
\end{equation}
Finally, because $\rho_{ABE}$ is pure, $\rho_{BE}$ and $\rho_A:=\mathrm{Tr}_{BE}(\rho_{ABE})$ have the same nonzero spectrum, hence $\mathrm{Tr}(\rho_{BE}^2)=\mathrm{Tr}(\rho_A^2)$. Since $\rho_A=\sigma_A=I_d/d$ (because $\sigma_{AB}\in\mathcal{F}_d$), we have $\mathrm{Tr}(\rho_A^2)=\mathrm{Tr}((I_d/d)^2)=1/d$, giving Eq.~\eqref{eq:purity-tradeoff-Fd}.
\end{proof}

\begin{theorem}[Closed form for $\upsilon_2(\Gamma_t)$]\label{thm:nu2-closed-form}
For every $d\ge 2$ and every $t$ in the CP-range,
\begin{equation}\label{eq:nu2-closed-form}
\upsilon_2(\Gamma_t)=\max\left\{
d^{-1/2},\ \ \|\tau_{AA'}\|_2 \right\}.
\end{equation}
Equivalently,
\begin{equation}\label{eq:nu2-explicit}
\upsilon_2(\Gamma_t) = \max\left\{d^{-1/2},\ \ \frac{1}{d}\sqrt{1+t^2(d^2-1)}\right\}.
\end{equation}
\end{theorem}
\begin{proof}
By Proposition~\ref{prop:choi-feasible-set},
\begin{equation}
\upsilon_2(\Gamma_t)
=
\sup_{\substack{B\ \mathrm{finite\mbox{ }dimension}\\
\sigma_{AB}\in\mathcal{F}_d(B)}}
\ \|(\Gamma_t\otimes \mathrm{id}_B)(\sigma_{AB})\|_2.
\end{equation}
Fix an admissible state $\sigma_{AB}$ and set
\begin{equation}
\omega_{AB}:=(\Gamma_t\otimes \mathrm{id}_B)(\sigma_{AB}),\qquad
a:=\mathrm{Tr}(\sigma_{AB}^2),\qquad
\beta:=\mathrm{Tr}(\sigma_B^2).
\end{equation}
Since $\sigma_{AB}$ and $\sigma_B$ are density operators, we have
\begin{equation}
0\le a\le 1,\qquad 0\le \beta\le 1.
\end{equation}
By Lemma~\ref{lem:p2-purity-identity},
\begin{equation}\label{eq:nu2-proof-start}
\|\omega_{AB}\|_2^2=t^2 a+\frac{1-t^2}{d}\,\beta.
\end{equation}
Moreover, Corollary~\ref{cor:purity-tradeoff-Fd} gives
\begin{equation}\label{eq:nu2-proof-tradeoff}
a+\beta\le 1+\frac{1}{d}.
\end{equation}
We distinguish two cases.

Case 1: $1-t^2(d+1)\ge 0$.
Using \eqref{eq:nu2-proof-tradeoff} to eliminate $a$, we obtain
\begin{align}
\|\omega_{AB}\|_2^2
&=t^2 a+\frac{1-t^2}{d}\,\beta \nonumber\\
&\le t^2\Bigl(1+\frac{1}{d}-\beta\Bigr)+\frac{1-t^2}{d}\,\beta \nonumber\\
&=t^2\Bigl(1+\frac{1}{d}\Bigr)+\frac{1-t^2(d+1)}{d}\,\beta \nonumber\\
&\le t^2\Bigl(1+\frac{1}{d}\Bigr)+\frac{1-t^2(d+1)}{d} = \frac{1}{d},
\end{align}
where in the last step we used $\beta\le 1$. Hence
\begin{equation}
\|\omega_{AB}\|_2\le d^{-1/2}.
\end{equation}

Case 2: $1-t^2(d+1)\le 0$.
Using \eqref{eq:nu2-proof-tradeoff} to eliminate $\beta$, we obtain
\begin{align}
\|\omega_{AB}\|_2^2
&= t^2 a+\frac{1-t^2}{d}\,\beta \nonumber\\
&\le t^2 a+\frac{1-t^2}{d}\Bigl(1+\frac{1}{d}-a\Bigr) \nonumber\\
&= \frac{1-t^2}{d}\Bigl(1+\frac{1}{d}\Bigr)
+\Bigl(t^2-\frac{1-t^2}{d}\Bigr)a \nonumber\\
&\le \frac{1-t^2}{d}\Bigl(1+\frac{1}{d}\Bigr)
+\Bigl(t^2-\frac{1-t^2}{d}\Bigr)
=
t^2+\frac{1-t^2}{d^2},
\end{align}
where in the last step we used $a\le 1$. Thus
\begin{equation}
\|\omega_{AB}\|_2 \le \frac{1}{d}\sqrt{1+t^2(d^2-1)}.
\end{equation}
Combining the two cases yields
\begin{equation}
\upsilon_2(\Gamma_t) \le \max\left\{d^{-1/2},\ \frac{1}{d}\sqrt{1+t^2(d^2-1)}\right\}.
\end{equation}
On the other hand, Proposition~\ref{cor:two-lower-bounds} gives the matching lower bounds
\begin{equation}
\upsilon_2(\Gamma_t)\ge d^{-1/2},
\qquad
\upsilon_2(\Gamma_t)\ge \|\tau_{AA'}\|_2.
\end{equation}
Finally, by Proposition~\ref{prop:tau-spectrum},
\begin{equation}
\|\tau_{AA'}\|_2^2=\mathrm{Tr}(\tau_{AA'}^2)=\frac{1+t^2(d^2-1)}{d^2},
\end{equation}
hence
\begin{equation}
\|\tau_{AA'}\|_2
=\frac{1}{d}\sqrt{1+t^2(d^2-1)}.
\end{equation}
Therefore the upper and lower bounds coincide, proving
\eqref{eq:nu2-closed-form} and \eqref{eq:nu2-explicit}.
\end{proof}

\begin{remark}[A simplification for $d\ge 3$]\label{rem:nu2-simplification}
If $d\ge 3$ and $t$ lies in the CP-range, then
\begin{equation}
\upsilon_2(\Gamma)=d^{-1/2}.
\end{equation}
\end{remark}

\begin{proof}
Since $t$ lies in the CP-range, we have
\begin{equation}
-\frac{1}{d-1}\le t\le \frac{1}{d+1},
\end{equation}
and hence
\begin{equation}
|t|\le \frac{1}{d-1}.
\end{equation}
Therefore
\begin{equation}
t^2(d+1)\le \frac{d+1}{(d-1)^2}\le 1,
\end{equation}
where the last inequality is equivalent to $d(d-3)\ge 0$, which holds for $d\ge 3$.
Now Theorem \ref{thm:nu2-closed-form} gives
\begin{equation}
\upsilon_2(\Gamma)=\max\left\{
d^{-1/2},\ \frac{1}{d}\sqrt{1+t^2(d^2-1)}
\right\}.
\end{equation}
Since
\begin{equation}
t^2(d+1)\le 1,
\end{equation}
we obtain
\begin{equation}
1+t^2(d^2-1)\le 1+(d-1)=d,
\end{equation}
because $d^2-1=(d-1)(d+1)$. Hence
\begin{equation}
\frac{1}{d}\sqrt{1+t^2(d^2-1)} \le \frac{1}{d}\sqrt{d} = d^{-1/2}.
\end{equation}
Thus the second term in the maximum is no larger than the first, and so
\begin{equation}
\upsilon_2(\Gamma)=d^{-1/2}.
\end{equation}
\end{proof}

\section{Multiplicativity of $\upsilon_2(\Gamma)$}\label{sec:multiplicity}

In this section, through a sequence of steps, we prove that $\upsilon_2(\Gamma_t)$ is multiplicative.

\begin{definition}[Multiplicativity at $p=2$]\label{def:multiplicativity-nu2}
We say that $\upsilon_2(\Gamma_t)$ is multiplicative (at $p=2$) if, for every $n\in\mathbb{N}$,
\begin{equation}\label{eq:multiplicativity-all-n}
\upsilon_2(\Gamma_t^{\otimes n})=\bigl(\upsilon_2(\Gamma_t)\bigr)^n,
\end{equation}
where $\upsilon_2(\Gamma_t^{\otimes n})$ denotes the $n$ copy quantity defined in Eq.~\eqref{eq:def-nu2-n} and it is defined in Eq.~\eqref{eq:def-nu2-n}.
\end{definition}

\subsection{Tensor-power formulation and multiplicativity problem at $p=2$}
Now we formulate the $n$-copy (tensor-power) analogue of $\upsilon_2(\Gamma_t)$ corresponding to $n$ parallel uses of the fixed map $\Gamma_t$, allowing an arbitrary joint CPTP post-processing channel across the $n$ copies, with arbitrary finite dimensional output space. The resulting quantity admits a natural multiplicativity question: whether entangled post-processing across copies can improve the optimum beyond product strategies.

For each $n\in\mathbb{N}$, let $A_1,\dots,A_n$ and $A'_1,\dots,A'_n$ be copies of $A\simeq\mathbb{C}^d$ and $A'\simeq\mathbb{C}^d$, and define
\begin{equation}
A^{(n)}:=A_1\otimes\cdots\otimes A_n,
\qquad
A'^{(n)}:=A'_1\otimes\cdots\otimes A'_n.
\end{equation}
All tensor product bases on $A^{(n)}$ and $A'^{(n)}$ are the product bases induced by the fixed single-copy basis $\{|i\rangle\}_{i=1}^d$. Let $\Phi_d\in\mathcal{L}(A\otimes A')$ be the normalized maximally entangled state defined earlier, and set
\begin{equation}
\Phi_d^{\otimes n}\in\mathcal{L}(A^{(n)}\otimes A'^{(n)}).
\end{equation}
Let $\Gamma_t:\mathcal{L}(A)\to\mathcal{L}(A)$ be the CPTP map defined earlier, and define its tensor power
\begin{equation}
\Gamma_t^{\otimes n}:\mathcal{L}(A^{(n)})\to\mathcal{L}(A^{(n)}).
\end{equation}

\begin{definition}[$n$-copy output operator and $n$-copy value at $p=2$]\label{def:n-copy-nu2}
For every finite dimensional output system $B^{(n)}$ and every CPTP map
\begin{equation}
\Lambda^{(n)}:\mathcal{L}(A'^{(n)})\to\mathcal{L}(B^{(n)}),
\end{equation}
define
\begin{equation}\label{eq:def-omega-n}
\omega^{(n)}_{A^{(n)}B^{(n)}}(\Lambda^{(n)})
:=
\bigl(\Gamma_t^{\otimes n}\otimes \id_{B^{(n)}}\bigr)
\Bigl(
\bigl(\id_{A^{(n)}}\otimes \Lambda^{(n)}\bigr)(\Phi_d^{\otimes n})
\Bigr)
\in \mathcal{L}(A^{(n)}\otimes B^{(n)}).
\end{equation}
We then define
\begin{equation}\label{eq:def-nu2-n}
\upsilon_2(\Gamma_t^{\otimes n}):=\sup_{\Lambda\ \mathrm{CPTP}}
\bigl\|\omega^{(n)}_{A^{(n)}B^{(n)}}(\Lambda^{(n)})\bigr\|_2.
\end{equation}
\end{definition}
\begin{lemma}[Multiplicativity of the Hilbert-Schmidt norm]\label{lem:HS-multiplicativity}
For all operators $X,Y$,
\begin{equation}
\|X\otimes Y\|_2=\|X\|_2\,\|Y\|_2.
\end{equation}
\end{lemma}

\begin{proof}
By definition and cyclicity of the trace,
\begin{equation}
\|X\otimes Y\|_2^2
=
\mathrm{Tr}\,\bigl((X^\dagger X)\otimes (Y^\dagger Y)\bigr)
=
\mathrm{Tr}\,(X^\dagger X)\,\mathrm{Tr}\,(Y^\dagger Y)
=
\|X\|_2^2\,\|Y\|_2^2.
\end{equation}
Taking square roots gives the claim.
\end{proof}

\begin{proposition}[Trivial supermultiplicativity]\label{prop:nu2-supermult}
For every $n\in\mathbb{N}$,
\begin{equation}\label{eq:nu2-supermult}
\upsilon_2(\Gamma_t^{\otimes n})\ \ge\ \bigl(\upsilon_2(\Gamma_t)\bigr)^n.
\end{equation}
\end{proposition}

\begin{proof}
Fix $\varepsilon>0$, and by the definition of $\upsilon_2(\Gamma_t)$, there exist a finite dimensional output system $B$ and a CPTP map $\Lambda:\mathcal{L}(A')\to\mathcal{L}(B)$ such that
\begin{equation}
\left\| \bigl(\Gamma_t\otimes \id_B\bigr)
\bigl((\id_A\otimes \Lambda)(\Phi_d)\bigr) \right\|_2
\ge \upsilon_2(\Gamma_t)-\varepsilon.
\end{equation}
Consider the product map
\begin{equation}
\Lambda^{\otimes n}:\mathcal{L}(A'^{(n)})\to\mathcal{L}(B^{\otimes n}),
\end{equation}
which is CPTP. Since $\Phi_d^{\otimes n}$ is a product state across copies and $\Gamma_t^{\otimes n}$ is a product map across copies,
\begin{equation}
\omega^{(n)}\bigl(\Lambda^{\otimes n}\bigr)
=
\Bigl(
\bigl(\Gamma_t\otimes \id_B\bigr)
\bigl((\id_A\otimes \Lambda)(\Phi_d)\bigr)
\Bigr)^{\otimes n}.
\end{equation}
By Lemma~\ref{lem:HS-multiplicativity},
\begin{equation}
\bigl\|\omega^{(n)}\bigl(\Lambda^{\otimes n}\bigr)\bigr\|_2
=
\left\|
\bigl(\Gamma_t\otimes \id_B\bigr)
\bigl((\id_A\otimes \Lambda)(\Phi_d)\bigr)
\right\|_2^{\,n}
\ge
\bigl(\upsilon_2(\Gamma_t)-\varepsilon\bigr)^n.
\end{equation}
Taking the supremum over all admissible $n$ copy channels yields
\begin{equation}
\upsilon_2(\Gamma_t^{\otimes n})\ge \bigl(\upsilon_2(\Gamma_t)-\varepsilon\bigr)^n.
\end{equation}
Since $\varepsilon>0$ was arbitrary, Eq.~\eqref{eq:nu2-supermult} follows.
\end{proof}

\begin{problem}[All-copy multiplicativity at $p=2$]\label{prob:nu2-multiplicativity}
Determine whether Eq.~\eqref{eq:multiplicativity-all-n} holds. Equivalently, determine whether there exist a finite dimensional output system $B^{(n)}$ and a joint CPTP map
\begin{equation}
\Lambda^{(n)}:\mathcal{L}(A'^{(n)})\to\mathcal{L}(B^{(n)})
\end{equation}
for which
\begin{equation}
\bigl\|\omega^{(n)}(\Lambda^{(n)})\bigr\|_2\ >\ \bigl(\upsilon_2(\Gamma_t)\bigr)^n,
\end{equation}
i.e.,\ whether joint post-processing across $n$ copies can outperform product post-processing.
\end{problem}
In the rest of the paper we prove that the multiplicativity holds.

\subsection{$n$-copy Choi slice reformulation at $p=2$}
Definition~\ref{def:n-copy-nu2} defines $\upsilon_2(\Gamma_t^{\otimes n})$ as a supremum over joint CPTP maps $\Lambda^{(n)}$ acting on $A'^{(n)}$, with arbitrary finite dimensional output space. As in the single copy case in last section, for each fixed output system $B^{(n)}$ this optimization can be expressed equivalently as a maximization over bipartite states whose $A^{(n)}$-marginal is maximally mixed. This reformulation isolates the admissible input set independently of channels and makes explicit the possibility of entangled (non-product) optimizers across $n$ copies.

\begin{definition}[$n$-copy Choi slice feasible set relative to $B^{(n)}$]
Let $B^{(n)}$ be a finite dimensional output system. Define
\begin{equation}\label{eq:def-F-n}
\mathcal{F}^{(n)}_d(B^{(n)}):=\Bigl\{\ \sigma_{A^{(n)}B^{(n)}}\in\mathcal{L}(A^{(n)}\otimes B^{(n)}):
\ \sigma_{A^{(n)}B^{(n)}}\ge 0,\ \rm{Tr}(\sigma_{A^{(n)}B^{(n)}})=1,\ 
\sigma_{A^{(n)}}=\frac{I_{d^n}}{d^n}\ \Bigr\}.
\end{equation}
\end{definition}

\begin{proposition}\label{prop:nu2n-choi-slice}
Let $B^{(n)}$ be a finite dimensional output system. Then the correspondence
\begin{equation}\label{eq:choi-map-n}
\Lambda^{(n)}\ \longmapsto\ \sigma_{A^{(n)}B^{(n)}}:=
(\id_{A^{(n)}}\otimes \Lambda^{(n)})(\Phi_d^{\otimes n})
\end{equation}
is an affine bijection from the set of CPTP maps
\begin{equation}
\Lambda^{(n)}:\mathcal{L}(A'^{(n)})\to\mathcal{L}(B^{(n)})
\end{equation}
onto $\mathcal{F}^{(n)}_d(B^{(n)})$. Consequently,
\begin{equation}\label{eq:nu2n-choi-slice}
\upsilon_2(\Gamma_t^{\otimes n})=\sup_{\substack{B^{(n)}\ }}
\ \max_{\sigma_{A^{(n)}B^{(n)}}\in\mathcal{F}^{(n)}_d(B^{(n)})}
\ \bigl\|(\Gamma_t^{\otimes n}\otimes \id_{B^{(n)}})(\sigma_{A^{(n)}B^{(n)}})\bigr\|_2.
\end{equation}
Equivalently,
\begin{equation}
\upsilon_2(\Gamma_t^{\otimes n})=\sup_{\substack{B^{(n)}\ \mathrm{finite\mbox{ }dimension}\\
\sigma_{A^{(n)}B^{(n)}}\in \mathcal{F}_d^{(n)}(B^{(n)})}}
\ \bigl\|(\Gamma_t^{\otimes n}\otimes \id_{B^{(n)}})(\sigma_{A^{(n)}B^{(n)}})\bigr\|_2.
\end{equation}
\end{proposition}

\begin{proof}
Fix a finite dimensional output system $B^{(n)}$.

Step 1: the image of a CPTP map lies in $\mathcal{F}^{(n)}_d(B^{(n)})$:
Let $\Lambda^{(n)}$ be CPTP and set $\sigma:=(\id_{A^{(n)}}\otimes \Lambda^{(n)})(\Phi_d^{\otimes n})$.
Since $\Phi_d^{\otimes n}\ge 0$ and $\id\otimes\Lambda^{(n)}$ is completely positive, $\sigma\ge 0$. Since $\id\otimes\Lambda^{(n)}$ is trace-preserving and $\rm{Tr}(\Phi_d^{\otimes n})=1$, we have $\rm{Tr}(\sigma)=1$. By using $\rm{Tr}_{A'^{(n)}}(\Phi_d^{\otimes n})=I_{d^n}/d^n$ and trace-preservation of $\Lambda^{(n)}$, we obtain
\begin{equation}
\sigma_{A^{(n)}}=\rm{Tr}_{B^{(n)}}(\sigma)=\rm{Tr}_{A'^{(n)}}(\Phi_d^{\otimes n}) =\frac{I_{d^n}}{d^n}.
\end{equation}
Thus $\sigma\in\mathcal{F}^{(n)}_d(B^{(n)})$.

Step 2: every $\sigma\in\mathcal{F}^{(n)}_d(B^{(n)})$ arises from a CPTP map: 
Let $\sigma\in\mathcal{F}^{(n)}_d(B^{(n)})$ and define the (unnormalized) Choi operator
\begin{equation}
J_{A^{(n)}B^{(n)}}:=d^n\,\sigma_{A^{(n)}B^{(n)}}.
\end{equation}
Define a linear map $\Lambda^{(n)}_\sigma:\mathcal{L}(A'^{(n)})\to\mathcal{L}(B^{(n)})$ by
\begin{equation}
\Lambda^{(n)}_\sigma(X):=\rm{Tr}_{A^{(n)}}\,\bigl((X^{T}\otimes I_{B^{(n)}})\,J_{A^{(n)}B^{(n)}}\bigr),
\end{equation}
where $T$ denotes transpose with respect to the fixed product basis on $A'^{(n)}$.
Since $J_{A^{(n)}B^{(n)}}\ge 0$, the map $\Lambda^{(n)}_\sigma$ is completely positive.
Moreover, using $\rm{Tr}_{B^{(n)}}(J)=d^n\,\sigma_{A^{(n)}}=I_{d^n}$, we obtain for all
$X\in\mathcal{L}(A'^{(n)})$,
\begin{equation}
\rm{Tr}\bigl(\Lambda^{(n)}_\sigma(X)\bigr) =\rm{Tr}\bigl((X^{T}\otimes I_{B^{(n)}})\,J\bigr) =\rm{Tr}\bigl(X^{T}\rm{Tr}_{B^{(n)}}(J)\bigr)=\rm{Tr}(X).
\end{equation}
Hence $\Lambda^{(n)}_\sigma$ is trace-preserving.
By construction of the Choi correspondence in the chosen basis,
\begin{equation}
(\id_{A^{(n)}}\otimes \Lambda^{(n)}_\sigma)(\Phi_d^{\otimes n})=\sigma.
\end{equation}
Step 3: bijectivity and the reformulation of $\upsilon_2(\Gamma_t^{\otimes n})$: 
In fact, steps 1-2 show that for the fixed output system $B^{(n)}$, the map ~\eqref{eq:choi-map-n}
is an affine bijection between CPTP maps $\Lambda^{(n)}:\mathcal{L}(A'^{(n)})\to\mathcal{L}(B^{(n)})$ and $\mathcal{F}^{(n)}_d(B^{(n)})$. Substituting $\sigma=(\id\otimes\Lambda^{(n)})(\Phi_d^{\otimes n})$ into definition \ref{def:n-copy-nu2} yields the stated reformulation after taking the supremum over all finite dimensional output systems $B^{(n)}$ yields
\eqref{eq:nu2n-choi-slice}.
\end{proof}

\begin{remark}[Attainment for fixed output space]
For each fixed finite dimensional output system $B^{(n)}$, the set $\mathcal{F}^{(n)}_d(B^{(n)})$ is compact and the map
\begin{equation}
\sigma\longmapsto
\|(\Gamma_t^{\otimes n}\otimes \id_{B^{(n)}})(\sigma)\|_2
\end{equation}
is continuous. Hence, for each fixed $B^{(n)}$, the maximum in
Eq.~\eqref{eq:nu2n-choi-slice} over $\mathcal{F}^{(n)}_d(B^{(n)})$ is attained; however, $\upsilon_2(\Gamma_t^{\otimes n})$ is an outer supremum over all finite dimension output spaces $B^{(n)}$.
\end{remark}

\subsection{A quadratic form expression for $\upsilon_2(\Gamma_t^{\otimes n})$}

In the $p=2$ case, the objective $\|(\Gamma_t^{\otimes n}\otimes \id_{B^{(n)}})(\sigma)\|_2$ can be squared and rewritten as a quadratic form in $\sigma$ involving the composed map
\begin{equation}
(\Gamma_t^2)^{\otimes n}.
\end{equation}
This is a useful simplification, since it removes the transpose from the norm expression and is the starting point for a systematic tensor product analysis.
\begin{definition}[The completely depolarizing map]\label{def:depolarizing-D}
Define $D:\mathcal{L}(A)\to\mathcal{L}(A)$ by
\begin{equation}\label{eq:def-D}
D(X):=\mathrm{Tr}(X)\,\frac{I_d}{d}.
\end{equation}
\end{definition}

\begin{lemma}[An algebraic identity for $\Gamma_t^2$]\label{lem:Gamma-square}
For every $t\in\mathbb{R}$,
\begin{equation}\label{eq:Gamma-square}
\Gamma_t \circ \Gamma_t:=\Gamma_t^2=t^2\,\id_A+(1-t^2)\,D.
\end{equation}
\end{lemma}

\begin{proof}
Write $\Gamma_t=t\,T+(1-t)\,D$, where $T(X)=X^T$ is the transpose map with respect to the fixed basis.
The maps $T$ and $D$ satisfy
\begin{equation}
T^2=\id_A,\qquad D^2=D,\qquad TD=DT=D.
\end{equation}
Indeed, $T^2=\id_A$ because transpose is an involution. Also,
\begin{equation}
D(D(X))=\mathrm{Tr}(X)\,D(I_d/d)
=\mathrm{Tr}(X)\,\frac{I_d}{d}=D(X),
\end{equation}
so $D^2=D$. Finally, $TD=DT=D$ because $D(X)$ is proportional to the identity and
$\mathrm{Tr}(X^T)=\mathrm{Tr}(X)$. Expanding the square gives
\begin{align*}
\Gamma_t^2&=\bigl(tT+(1-t)D\bigr)\bigl(tT+(1-t)D\bigr)\\
&=t^2\,\id_A+\bigl(2t(1-t)+(1-t)^2\bigr)D\\
&=t^2\,\id_A+(1-t^2)\,D,
\end{align*}
which is Eq.~\eqref{eq:Gamma-square}.
\end{proof}

\begin{lemma}[Self-adjoint property of $\Gamma_t$]\label{lem:Gamma-selfadjoint}
For all $X,Y\in\mathcal{L}(A)$,
\begin{equation}\label{eq:Gamma-selfadjoint}
\mathrm{Tr}\,\bigl(Y\,\Gamma_t(X)\bigr)=\mathrm{Tr}\,\bigl(\Gamma_t(Y)\,X\bigr).
\end{equation}
Consequently, for every finite dimensional system $B$ and all
$X_{AB},Y_{AB}\in\mathcal{L}(A\otimes B)$,
\begin{equation}\label{eq:Gamma-tensor-selfadjoint}
\mathrm{Tr}\bigl(Y_{AB}\,(\Gamma_t\otimes \id_B)(X_{AB})\bigr)
=\mathrm{Tr}\bigl((\Gamma_t\otimes \id_B)(Y_{AB})\,X_{AB}\bigr).
\end{equation}
\end{lemma}

\begin{proof}
Using
\begin{equation}
\Gamma_t(X)=tX^T+(1-t)\,\mathrm{Tr}(X)\frac{I_d}{d},
\end{equation}
we compute
\begin{align*}
\mathrm{Tr}\bigl(Y\,\Gamma_t(X)\bigr)&=t\,\mathrm{Tr}(Y X^T)
+(1-t)\,\mathrm{Tr}(X)\,\mathrm{Tr}\,\left(Y\frac{I_d}{d}\right)\\
&=t\,\mathrm{Tr}(Y^T X) +(1-t)\,\mathrm{Tr}(Y)\,\mathrm{Tr}\,\left(X\frac{I_d}{d}\right)\\
&=\mathrm{Tr}\bigl(\Gamma_t(Y)\,X\bigr),
\end{align*}
which proves Eq.~\eqref{eq:Gamma-selfadjoint}. Further, applying this identity on the $A$-factor and using linearity gives Eq.~\eqref{eq:Gamma-tensor-selfadjoint}.
\end{proof}
\begin{proposition}[Quadratic form]\label{prop:nu2n-quadratic-form}
Let $n\in\mathbb{N}$, let $B^{(n)}$ be a finite dimensional output system, and let
\begin{equation}
\sigma_{A^{(n)}B^{(n)}}\in \mathcal{F}^{(n)}_d(B^{(n)}).
\end{equation}
Define
\begin{equation}
\omega_{A^{(n)}B^{(n)}}:=
(\Gamma_t^{\otimes n}\otimes \id_{B^{(n)}})(\sigma_{A^{(n)}B^{(n)}}).
\end{equation}
Then
\begin{equation}\label{eq:nu2n-quadratic-form}
\|\omega_{A^{(n)}B^{(n)}}\|_2^2=\mathrm{Tr}\,\left(
\sigma_{A^{(n)}B^{(n)}}\,
\bigl((\Gamma_t^2)^{\otimes n}\otimes \id_{B^{(n)}}\bigr)(\sigma_{A^{(n)}B^{(n)}})
\right).
\end{equation}
Consequently,
\begin{equation}\label{eq:nu2n-squared-variational}
\bigl(\upsilon_2(\Gamma_t^{\otimes n})\bigr)^2
=\sup_{\substack{B^{(n)}\ }}
\ \max_{\sigma_{A^{(n)}B^{(n)}}\in\mathcal{F}^{(n)}_d(B^{(n)})}
\mathrm{Tr}\,\left(\sigma_{A^{(n)}B^{(n)}}\,
\bigl((\Gamma_t^2)^{\otimes n}\otimes \id_{B^{(n)}}\bigr)(\sigma_{A^{(n)}B^{(n)}})
\right).
\end{equation}
\end{proposition}

\begin{proof}
Set
\begin{equation}
L:=\Gamma_t^{\otimes n}\otimes \id_{B^{(n)}}.
\end{equation}
Since $\sigma\ge 0$ and $L$ is completely positive, $L(\sigma)\ge 0$. Hence
$\omega=L(\sigma)$ is positive semidefinite, and therefore
\begin{equation}
\|\omega\|_2^2=\mathrm{Tr}(\omega^2)=\mathrm{Tr}\bigl(L(\sigma)\,L(\sigma)\bigr).
\end{equation}
By Lemma \ref{lem:Gamma-selfadjoint}, $\Gamma_t$ is self-adjoint with respect to the
Hilbert-Schmidt inner product, and therefore so is $L$:
\begin{equation}
\mathrm{Tr}\bigl(L(X)\,Y\bigr)=\mathrm{Tr}\bigl(X\,L(Y)\bigr)
\qquad \text{for all }X,Y\in\mathcal{L}(A^{(n)}\otimes B^{(n)}).
\end{equation}
Applying this with $X=\sigma$ and $Y=L(\sigma)$ yields
\begin{equation}
\mathrm{Tr}\bigl(L(\sigma)\,L(\sigma)\bigr)=\mathrm{Tr}\bigl(\sigma\,L(L(\sigma))\bigr).
\end{equation}
Finally,
\begin{equation}
L\circ L =(\Gamma_t^{\otimes n}\circ\Gamma_t^{\otimes n})\otimes \id_{B^{(n)}}
= (\Gamma_t^2)^{\otimes n}\otimes \id_{B^{(n)}},
\end{equation}
which gives Eq.~\eqref{eq:nu2n-quadratic-form}.

For each fixed output system $B^{(n)}$, maximizing over
$\sigma\in\mathcal{F}^{(n)}_d(B^{(n)})$ gives the corresponding fixed output variational formula.
Taking the supremum over all finite dimensional output systems $B^{(n)}$ yields
Eq.~\eqref{eq:nu2n-squared-variational}.
\end{proof}

\begin{corollary}[A transpose-free representation]\label{cor:nu2n-transpose-free}
For every $n\in\mathbb{N}$,
\begin{equation}\label{eq:nu2n-transpose-free}
\bigl(\upsilon_2(\Gamma_t^{\otimes n})\bigr)^2
=\sup_{\substack{B^{(n)}\ }}
\ \max_{\sigma\in\mathcal{F}^{(n)}_d(B^{(n)})}
\mathrm{Tr}\,\left(\sigma\,\bigl(\bigl(t^2\,\id_A+(1-t^2)D\bigr)^{\otimes n}\otimes \id_{B^{(n)}}\bigr)(\sigma)
\right),
\end{equation}
where $D$ is the completely depolarizing map from Definition~\ref{def:depolarizing-D}.
\end{corollary}

\begin{proof}
Substitution of Eq.~\eqref{eq:Gamma-square} from Lemma \ref{lem:Gamma-square} into
Eq.~\eqref{eq:nu2n-squared-variational} yields the claim.
\end{proof}

\subsection{Subset expansion and a marginal-purity identity}

Corollary \ref{cor:nu2n-transpose-free} reduces the $n$-copy $p=2$ objective to a quadratic form involving
\begin{equation}
\bigl(t^2\,\id_A+(1-t^2)D\bigr)^{\otimes n}.
\end{equation}
The present subsection expands this tensor power as a sum over subsets $S\subseteq\{1,\dots,n\}$ and shows that, when inserted into the quadratic form, each subset term evaluates to a normalized purity of a marginal of the optimizing state $\sigma\in\mathcal{F}_d^{(n)}$. This yields an exact subset-purity formula for $\bigl(\upsilon_2(\Gamma_t^{\otimes n})\bigr)^2$.

For $n\in\mathbb{N}$ we write $[n]:=\{1,2,\dots,n\}$. For $S\subseteq[n]$, write $S^c:=[n]\setminus S$ and
\begin{equation}
A_S:=\bigotimes_{k\in S} A_k,\qquad A_{S^c}:=\bigotimes_{k\in S^c} A_k,
\qquad \text{so that }A^{(n)}\simeq A_S\otimes A_{S^c}.
\end{equation}
For $\sigma_{A^{(n)}B^{(n)}}\in\mathcal{L}(A^{(n)}\otimes B^{(n)})$, define the reduced operator
\begin{equation}\label{eq:def-sigma-ASB}
\sigma_{A_SB^{(n)}}:=\mathrm{Tr}_{A_{S^c}}(\sigma_{A^{(n)}B^{(n)}})\ \in\ \mathcal{L}(A_S\otimes B^{(n)}).
\end{equation}
\begin{lemma}[Subset expansion of $\bigl(t^2\id+(1-t^2)D\bigr)^{\otimes n}$]\label{lem:subset-expansion}
For every $n\in\mathbb{N}$,
\begin{equation}\label{eq:subset-expansion}
\bigl(t^2\,\id_A+(1-t^2)D\bigr)^{\otimes n}
=\sum_{S\subseteq[n]} t^{2|S|}(1-t^2)^{|S^c|}
\ \Bigl(\id_{A_S}\otimes D_{A_{S^c}}\Bigr),
\end{equation}
where $\id_{A_S}$ denotes the identity map on $\mathcal{L}(A_S)$ and
$D_{A_{S^c}}$ denotes the tensor product of $D$ on the factors indexed by $S^c$.
\end{lemma}

\begin{proof}
This is the distributive expansion of the $n$-fold tensor product
\begin{equation}
\bigotimes_{k=1}^n \bigl(t^2\,\id_{A_k}+(1-t^2)D_{A_k}\bigr).
\end{equation}
For each subset $S\subseteq[n]$, select the term $t^2\,\id_{A_k}$ for $k\in S$ and the term $(1-t^2)D_{A_k}$ for $k\in S^c$; the product of the coefficients is $t^{2|S|}(1-t^2)^{|S^c|}$ and the corresponding tensor-product map is $\id_{A_S}\otimes D_{A_{S^c}}$. Summing over all $S$ yields
Eq.~\eqref{eq:subset-expansion}.
\end{proof}

\begin{lemma}[Marginal purity identity]\label{lem:marginal-purity-identity}
Let $S\subseteq[n]$ and let $\sigma_{A^{(n)}B^{(n)}}\in\mathcal{L}(A^{(n)}\otimes B^{(n)})$ be arbitrary.
Then
\begin{equation}\label{eq:marginal-purity-identity}
\mathrm{Tr}\,\left( \sigma_{A^{(n)}B^{(n)}} \, \bigl((\id_{A_S}\otimes D_{A_{S^c}})\otimes \id_{B^{(n)}}\bigr)(\sigma_{A^{(n)}B^{(n)}})\right)= \frac{1}{d^{|S^c|}}\,
\mathrm{Tr}\,\left(\sigma_{A_SB^{(n)}}^2\right),
\end{equation}
where $\sigma_{A_SB^{(n)}}$ is defined in Eq.~\eqref{eq:def-sigma-ASB}.
\end{lemma}
\begin{proof}
Set $\rho:=\sigma_{A_SB^{(n)}}=\mathrm{Tr}_{A_{S^c}}(\sigma_{A^{(n)}B^{(n)}})$ and let $S\subseteq[n]$ and let $X\in\mathcal{L}(A^{(n)})\simeq\mathcal{L}(A_S\otimes A_{S^c})$. By definition of $D$ and of the tensor product map $D_{A_{S^c}}$, we have
\begin{equation}
(\id_{A_S}\otimes D_{A_{S^c}})(X_S\otimes X_{S^c})
=X_S\otimes D_{A_{S^c}}(X_{S^c})=X_S\otimes \mathrm{Tr}(X_{S^c})\,\frac{I_{A_{S^c}}}{d^{|S^c|}}.
\end{equation}
On the other hand, $\mathrm{Tr}_{A_{S^c}}(X_S\otimes X_{S^c})=\mathrm{Tr}(X_{S^c})\,X_S$, so
\begin{equation}
\mathrm{Tr}_{A_{S^c}}(X_S\otimes X_{S^c})\otimes \frac{I_{A_{S^c}}}{d^{|S^c|}}
=\mathrm{Tr}(X_{S^c})\,X_S\otimes \frac{I_{A_{S^c}}}{d^{|S^c|}},
\end{equation}
Therefore we obtain
\begin{equation}\label{eq:trace-and-replace}
(\id_{A_S}\otimes D_{A_{S^c}})(X)
=\mathrm{Tr}_{A_{S^c}}(X)\ \otimes\ \frac{I_{A_{S^c}}}{d^{|S^c|}}.
\end{equation}
where applying this to the $A$-registers and leaving $B^{(n)}$ untouched, yields
\begin{equation}
\bigl((\id_{A_S}\otimes D_{A_{S^c}})\otimes \id_{B^{(n)}}\bigr)(\sigma)
=\rho_{A_SB^{(n)}}\otimes \frac{I_{A_{S^c}}}{d^{|S^c|}},
\end{equation}
after the canonical identification
\begin{equation}
A^{(n)}\otimes B^{(n)}
\cong A_S\otimes A_{S^c}\otimes B^{(n)} \cong A_S\otimes B^{(n)}\otimes A_{S^c}.
\end{equation}
Therefore,
\begin{equation}
\mathrm{Tr}\,\left(\sigma\,
\bigl((\id_{A_S}\otimes D_{A_{S^c}})\otimes \id_{B^{(n)}}\bigr)(\sigma)
\right)=\frac{1}{d^{|S^c|}}\,
\mathrm{Tr}\,\left(\sigma\,(\rho\otimes I_{A_{S^c}})\right).
\end{equation}
Since $\mathrm{Tr}_{A_{S^c}}$ is the Hilbert-Schmidt adjoint of the map
$X\mapsto X\otimes I_{A_{S^c}}$, we get
\begin{equation}
\mathrm{Tr}\,\left(\sigma\,(\rho\otimes I_{A_{S^c}})\right)
=\mathrm{Tr}\,\left(\mathrm{Tr}_{A_{S^c}}(\sigma)\,\rho\right)
=\mathrm{Tr}(\rho^2),
\end{equation}
which proves Eq.~\eqref{eq:marginal-purity-identity}.
\end{proof}

\begin{proposition}[Subset purity formula for $\bigl(\upsilon_2(\Gamma_t^{\otimes n})\bigr)^2$]\label{prop:subset-purity-formula}
For every $n\in\mathbb{N}$,
\begin{equation}\label{eq:subset-purity-formula}
\bigl(\upsilon_2(\Gamma_t^{\otimes n})\bigr)^2 =\sup_{\substack{B^{(n)}\ }}
\ \max_{\sigma\in\mathcal{F}^{(n)}_d(B^{(n)})}
\ \sum_{S\subseteq[n]}
t^{2|S|}(1-t^2)^{|S^c|} \ \frac{1}{d^{|S^c|}}\,
\mathrm{Tr}\,\left(\sigma_{A_SB^{(n)}}^2\right).
\end{equation}
\end{proposition}

\begin{proof}
Start from the transpose-free representation in Corollary~\ref{cor:nu2n-transpose-free}.
For each fixed finite dimensional output system $B^{(n)}$, expand
\begin{equation}
\bigl(t^2\id+(1-t^2)D\bigr)^{\otimes n}
\end{equation}
using Lemma~\ref{lem:subset-expansion}, and apply Lemma~\ref{lem:marginal-purity-identity} term by term to convert each quadratic-form term into the corresponding normalized marginal purity. This yields, for each fixed $B^{(n)}$,
\begin{equation}
\max_{\sigma\in\mathcal{F}^{(n)}_d(B^{(n)})}
\sum_{S\subseteq[n]} t^{2|S|}(1-t^2)^{|S^c|} \frac{1}{d^{|S^c|}}
\mathrm{Tr}\,\left(\sigma_{A_SB^{(n)}}^2\right).
\end{equation}
Taking the supremum over all finite dimensional output systems $B^{(n)}$ gives
Eq.~\eqref{eq:subset-purity-formula}.
\end{proof}


\begin{theorem}[All-$n$ multiplicativity at $p=2$]\label{thm:p2-all-n-multiplicativity} For every $d\ge 2$, every $t$ in the CP-range, and every $n\in\mathbb{N}$,
\begin{equation}\label{eq:p2-all-n-mult}
\upsilon_2(\Gamma_t^{\otimes n})=\bigl(\upsilon_2(\Gamma_t)\bigr)^n.
\end{equation}
Moreover, we have
\begin{equation}\label{eq:nu2-onecopy-m}
\upsilon_2(\Gamma_t)=\sqrt{\max\left\{\frac{1}{d},\ t^2+\frac{1-t^2}{d^2}\right\}},
\end{equation}
\end{theorem}
\begin{proof}
Fix $n\in\mathbb{N}$, let $B^{(n)}$ be an arbitrary finite dimensional output system, and let
\begin{equation}
\sigma_{A^{(n)}B^{(n)}}\in \mathcal{F}^{(n)}_d(B^{(n)}).
\end{equation}

\smallskip
\noindent{ Step 1: Swap polynomial representation.}
Let $\widetilde{A}_1,\dots,\widetilde{A}_n$ be copies of $A_1,\dots,A_n$, and let
$\widetilde{B}^{(n)}$ be a copy of $B^{(n)}$. View $\sigma\otimes\sigma$ as an operator on {the doubled space}
\begin{equation}
\bigl(A^{(n)}\otimes B^{(n)}\bigr)\otimes
\bigl(\widetilde{A}^{(n)}\otimes \widetilde{B}^{(n)}\bigr),
\end{equation}
For each $k=1,\dots,n$, let
\begin{equation}
U_k:=\Pi_{A_k,\widetilde A_k}
\end{equation}
denote the swap on $A_k\otimes \widetilde A_k$; extended by the identity on all other tensor factors. the same
\begin{equation}
V:=\Pi_{B^{(n)},\widetilde B^{(n)}}
\end{equation}
denote the swap on $B^{(n)}\otimes \widetilde B^{(n)}$; extended by the identity on
$A^{(n)}\otimes \widetilde A^{(n)}$. Then
\begin{equation}
U_k^2=V^2=I,
\end{equation}
and the operators $U_1,\dots,U_n,V$ commute pairwise. Now  fix $S\subseteq[n]$ and write
\begin{equation}
\sigma_{A_SB^{(n)}}:=\operatorname{Tr}_{A_{S^c}}(\sigma).
\end{equation}
By the swap trick applied to the reduced state $\sigma_{A_SB^{(n)}}$,
\begin{equation}
\operatorname{Tr}\bigl(\sigma_{A_SB^{(n)}}^2\bigr)=\operatorname{Tr}\Bigl(
(\sigma_{A_SB^{(n)}}\otimes \sigma_{\widetilde A_S\widetilde B^{(n)}})
\bigl(\Pi_{A_S,\widetilde A_S}\otimes \Pi_{B^{(n)},\widetilde B^{(n)}}\bigr)
\Bigr).
\end{equation}
Equivalently, let $\widehat \Pi_S$ denote the operator on the full doubled space which acts as
\begin{equation}
 \Pi_{A_S,\widetilde A_S}\otimes \Pi_{B^{(n)},\widetilde B^{(n)}}
\end{equation}
on the factors indexed by $A_S,\widetilde A_S,B^{(n)},\widetilde B^{(n)}$, and as the identity on the omitted factors $A_{S^c}\otimes \widetilde A_{S^c}$. Then
\begin{equation}
\operatorname{Tr}\bigl(\sigma_{A_SB^{(n)}}^2\bigr)=\operatorname{Tr}\Bigl((\sigma\otimes \sigma)\,\widehat \Pi_S\Bigr).
\end{equation}
By construction,
\begin{equation}
\widehat \Pi_S=\Bigl(\prod_{k\in S}U_k\Bigr)V=V\Bigl(\prod_{k\in S}U_k\Bigr),
\end{equation}
where the product is the ordinary operator product of commuting swaps. Therefore,
\begin{equation}\label{eq:subset-swap-identity}
\operatorname{Tr}\bigl(\sigma_{A_SB^{(n)}}^2\bigr)=\operatorname{Tr}\Bigl((\sigma\otimes \sigma)\,V\prod_{k\in S}U_k\Bigr).
\end{equation}
Substituting Eq.~\eqref{eq:subset-swap-identity} into Proposition \ref{prop:subset-purity-formula}, we obtain
\begin{align}
\bigl\|(\Gamma_t^{\otimes n}\otimes \id_{B^{(n)}})(\sigma)\bigr\|_2^2
&=\sum_{S\subseteq[n]}t^{2|S|}(1-t^2)^{|S^c|}\frac{1}{d^{|S^c|}}
\operatorname{Tr}\Bigl((\sigma\otimes \sigma)\,V\prod_{k\in S}U_k\Bigr)\nonumber\\
&=\operatorname{Tr}\Bigl((\sigma\otimes \sigma)\,V\sum_{S\subseteq[n]}
t^{2|S|}(1-t^2)^{|S^c|}\frac{1}{d^{|S^c|}}\prod_{k\in S}U_k\Bigr)\nonumber\\
&=\operatorname{Tr}\Bigl((\sigma\otimes \sigma)\,V\prod_{k=1}^n\Bigl(t^2\,U_k+\frac{1-t^2}{d}\,I\Bigr)\Bigr).
\label{eq:fn-swap-poly}
\end{align}
Thus it suffices to bound the right hand side of Eq.~\eqref{eq:fn-swap-poly} uniformly over all finite dimensional output systems $B^{(n)}$ and all
$\sigma\in\mathcal{F}^{(n)}_d(B^{(n)})$.

Recall that $U_1,\dots,U_n,V$ are commuting self-adjoint involutions.
Set
\begin{equation}
    a:=\frac{1-t^2}{d},
    \qquad
    G:=\prod_{k=1}^n (aI+t^2U_k).
\end{equation}
Then $G$ is self-adjoint and commutes with $V$. Since $V$ and $G$ are commuting self-adjoint operators, hence we obtain
\begin{equation}\label{eq:VG-absG}
    VG\le |G|,
\end{equation}
where $|G|=(G^2)^{1/2}$ and the inequality is in the Loewner order. Indeed, on any joint eigenspace of $V$ and $G$, the operator $V$ has eigenvalue
$\eta\in\{\pm1\}$ and $G$ has some real eigenvalue $g$, and therefore
\begin{equation}
    \eta g\le |g|.
\end{equation}
Equivalently, $|G|-VG$ is positive semi-definite. Because $\sigma\otimes\sigma\ge0$, taking the trace of Eq.~\eqref{eq:VG-absG} against $\sigma\otimes\sigma$ gives
\begin{equation}\label{eq:first-upper-bound}
\Tr\bigl((\sigma\otimes\sigma)VG\bigr) \le
\Tr\bigl((\sigma\otimes\sigma)|G|\bigr).
\end{equation}
It remains to evaluate the right-hand side. Since the factors $aI+t^2U_k$ commute and are self-adjoint, hence we obtain
\begin{equation}
    |G|=\biggl|\prod_{k=1}^n(aI+t^2U_k)\biggr|
    =\prod_{k=1}^n |aI+t^2U_k|.
\end{equation}
For a single swap $U_k$, whose eigenvalues are $\pm1$, write
\begin{equation}
    |aI+t^2U_k|=\alpha I+\beta U_k,
\end{equation}
where
\begin{equation}
    \alpha:=\frac{|a+t^2|+|a-t^2|}{2},
    \qquad
    \beta:=\frac{|a+t^2|-|a-t^2|}{2}.
\end{equation}
Hence
\begin{equation}
    |G|=\prod_{k=1}^n(\alpha I+\beta U_k)
    =\sum_{T\subseteq[n]}\alpha^{n-|T|}\beta^{|T|}\prod_{k\in T}U_k.
\end{equation}
Therefore
\begin{align}
\Tr\bigl((\sigma\otimes\sigma)|G|\bigr)
&=\sum_{T\subseteq[n]}\alpha^{n-|T|}\beta^{|T|}
\Tr\Bigl((\sigma\otimes\sigma)\prod_{k\in T}U_k\Bigr). \label{eq:absG-expanded}
\end{align}
By the swap trick applied to the marginal $\sigma_{A_T}$,
\begin{equation}
    \Tr\Bigl((\sigma\otimes\sigma)\prod_{k\in T}U_k\Bigr)
    =\Tr(\sigma_{A_T}^2).
\end{equation}
Since $\sigma_{A^{(n)}}=I_{d^n}/d^n$, we have
\begin{equation}
    \sigma_{A_T}=\frac{I_{d^{|T|}}}{d^{|T|}},
    \qquad
    \Tr(\sigma_{A_T}^2)=d^{-|T|}.
\end{equation}
Substituting this into Eq.~\eqref{eq:absG-expanded}, we obtain
\begin{align}
\Tr\bigl((\sigma\otimes\sigma)|G|\bigr)
&=
\sum_{T\subseteq[n]}\alpha^{n-|T|}\beta^{|T|}d^{-|T|}=
\left(\alpha+\frac{\beta}{d}\right)^n. \label{eq:alpha-beta-bound}
\end{align}
Now
\begin{align}
\alpha+\frac{\beta}{d}
&=
\frac{1+\frac1d}{2}\,|a+t^2|
+
\frac{1-\frac1d}{2}\,|a-t^2|. \label{eq:one-bit-quantity}
\end{align}
Recalling that $a=(1-t^2)/d$, we have $a+t^2\ge0$. Moreover,
\begin{equation}
    a-t^2=\frac{1-(d+1)t^2}{d}.
\end{equation}
Thus, if $t^2\le 1/(d+1)$, then $a-t^2\ge0$ and Eq.~\eqref{eq:one-bit-quantity} becomes
\begin{equation}
\frac{1+\frac1d}{2}(a+t^2)
+
\frac{1-\frac1d}{2}(a-t^2)
=
\frac1d.
\end{equation}
If $t^2\ge 1/(d+1)$, then $a-t^2\le0$ and Eq.~\eqref{eq:one-bit-quantity} in this case becomes
\begin{equation}
\frac{1+\frac1d}{2}(a+t^2)
+
\frac{1-\frac1d}{2}(t^2-a)
=t^2+\frac{1-t^2}{d^2}.
\end{equation}
Consequently,
\begin{equation}
    \alpha+\frac{\beta}{d}
    =
    \max\left\{\frac1d,\,t^2+\frac{1-t^2}{d^2}\right\}.
\end{equation}
Considering also Eq.~\eqref{eq:fn-swap-poly}, Eq.~\eqref{eq:first-upper-bound}, and Eq.~\eqref{eq:alpha-beta-bound}, we conclude that

\begin{equation}
\bigl\|(\Gamma_t^{\otimes n}\otimes \id_{B^{(n)}})(\sigma)\bigr\|_2^2
\le
\left(\max\left\{\frac1d,\,t^2+\frac{1-t^2}{d^2}\right\}\right)^n
= m(d,t)^n.
\end{equation}
Since this bound holds for every finite dimensional output system $B^{(n)}$ and every admissible $\sigma\in \mathcal{F}_d^{(n)}(B^{(n)})$, taking the supremum yields 
\begin{equation}\label{eq:upper-bound-mn}
\bigl(\upsilon_2(\Gamma_t^{\otimes n})\bigr)^2\ \le\ (m(d,t))^n.
\end{equation}
In what follows, we obtain matching lower bound by explicit product optimizers.
Choose the one copy output system to be
\begin{equation}
B_\star:=A'.
\end{equation}
Define a one copy admissible state $\rho^\star\in \mathcal{F}^{(1)}_d(B_\star)$ as follows. 

Case 1: If $t^2\le \frac{1}{d+1},$ fix a unit vector $|\psi\rangle\in A'$ and set
\begin{equation}
\rho^\star:=\frac{I_d}{d}\otimes |\psi\rangle\langle\psi|.
\end{equation}
Then
\begin{equation}
(\Gamma_t\otimes \id_{B_\star})(\rho^\star) =\frac{I_d}{d}\otimes |\psi\rangle\langle\psi|,
\end{equation}
and therefore
\begin{equation}
\bigl\|(\Gamma_t\otimes \id_{B_\star})(\rho^\star)\bigr\|_2^2
=\left\|\frac{I_d}{d}\right\|_2^2=\frac{1}{d}.
\end{equation}
Case 2: If $t^2\ge \frac{1}{d+1},$ then set
\begin{equation}
\rho^\star:=\Phi_d,
\end{equation}
viewed as a state on $A\otimes A'=A\otimes B_\star$. Then
\begin{equation}
(\Gamma_t\otimes \id_{B_\star})(\rho^\star)=\tau_{AA'},
\end{equation}
so
\begin{equation}
\bigl\|(\Gamma_t\otimes \id_{B_\star})(\rho^\star)\bigr\|_2^2
=\|\tau_{AA'}\|_2^2 =\frac{1+t^2(d^2-1)}{d^2}=t^2+\frac{1-t^2}{d^2}.
\end{equation}
In either case,
\begin{equation}
\bigl\|(\Gamma_t\otimes \id_{B_\star})(\rho^\star)\bigr\|_2^2=m(d,t).
\end{equation}
Now define
\begin{equation}
\sigma^\star:=\bigl(\rho^\star\bigr)^{\otimes n}\in
\mathcal{F}^{(n)}_d\bigl(B_\star^{\otimes n}\bigr).
\end{equation}
By multiplicativity of the Hilbert-Schmidt norm,
\begin{equation}
\bigl\|(\Gamma_t^{\otimes n}\otimes \id_{B_\star^{\otimes n}})(\sigma^\star)\bigr\|_2^2=\left( \bigl\|(\Gamma_t\otimes \id_{B_\star})(\rho^\star)\bigr\|_2^2 \right)^n= (m(d,t))^n.
\end{equation}
Since $B_\star^{\otimes n}$ is an admissible output system, it follows that
\begin{equation}\label{eq:lower-bound-mn}
\bigl(\upsilon_2(\Gamma_t^{\otimes n})\bigr)^2\ \ge\ (m(d,t))^n.
\end{equation}
Combining Eq.~\eqref{eq:upper-bound-mn} and Eq.~\eqref{eq:lower-bound-mn} yields Eq.~\eqref{eq:p2-all-n-mult}.

\end{proof}

 \section{A General Criterion for Multiplicativity}\label{sec:general-criterion}

\begin{proposition}[A general $p=2$ multiplicativity criterion]\label{prop:general-p2-criterion}
Let $A\simeq \mathbb C^d$, let $C$ be a finite dimensional output system, and let
\begin{equation}
\mathcal N:\mathcal L(A)\to \mathcal L(C)
\end{equation}
be a CP map. Assume that there exist constants $a,b\ge 0$ such that
\begin{equation}\label{eq:general-N-assumption-ab}
\mathcal N^\dagger\circ \mathcal N=a\,\id_A+b\,\Tr[\cdot]\,I_d .
\end{equation}
Then
\begin{equation}\label{eq:N-onecopy-closed-ab}
\upsilon_{2}(\mathcal N)=\sqrt{\max\left\{a+\frac{b}{d},\ \frac{a}{d}+b\right\}}, 
\end{equation}
and  for every $n\in\mathbb N$,
\begin{equation}\label{eq:N-ncopy-closed-ab}
 \upsilon_{2}(\cN^{\otimes n})= \bigl(\upsilon_{2}(\cN) \bigr)^n.  
\end{equation}
\end{proposition}

\begin{proof}
We first consider the one copy problem. Let $\sigma_{AB}$ be admissible, i.e.
\begin{equation}
\sigma_{AB}\ge 0,\qquad \Tr(\sigma_{AB})=1,\qquad \sigma_A=\frac{I_d}{d}.
\end{equation}
Now set
\begin{equation}
\omega_{CB}:=(\mathcal N\otimes \id_B)(\sigma_{AB}).
\end{equation}
Since $\omega_{CB}$ is positive semidefinite, we have
\begin{equation}
\|\omega_{CB}\|_2^2=\Tr(\omega_{CB}^2).
\end{equation}
By using the Hilbert-Schmidt adjoint, we obtain
\begin{align}
\|\omega_{CB}\|_2^2&=\Tr\Bigl(
\sigma_{AB}\, \bigl((\mathcal N^\dagger\circ \mathcal N)\otimes \id_B\bigr)(\sigma_{AB}) \Bigr)\nonumber\\
&=\Tr\Bigl(\sigma_{AB}\,\bigl((a\,\id_A+b\,\Tr[\cdot]\,I_d)\otimes \id_B\bigr)(\sigma_{AB})\Bigr)\nonumber\\
&=a\,\Tr(\sigma_{AB}^2)+b\,\Tr(\sigma_B^2).
\label{eq:general-N-purity-identity-ab}
\end{align}
By Corollary \ref{cor:purity-tradeoff-Fd},
\begin{equation}\label{eq:general-N-tradeoff-ab}
\Tr(\sigma_{AB}^2)+\Tr(\sigma_B^2)\le 1+\frac{1}{d}.
\end{equation}
If $a\ge b$, we eliminate $\Tr(\sigma_B^2)$ from Eq.~\eqref{eq:general-N-purity-identity-ab} and obtain
\begin{align}
\|\omega_{CB}\|_2^2 & \le a\,\Tr(\sigma_{AB}^2)+b\Bigl(1+\frac{1}{d}-\Tr(\sigma_{AB}^2)\Bigr)\nonumber\\
&=\Bigl(a-b\Bigr)\Tr(\sigma_{AB}^2)+b\Bigl(1+\frac{1}{d}\Bigr)\nonumber\\
& \le a+b \Bigl(1+\frac{1}{d}\Bigr)-b=a+\frac{b}{d},
\end{align}
where we used $\Tr(\sigma_{AB}^2)\le 1$.

If $b\ge a$, we eliminate $\Tr(\sigma_{AB}^2)$:
\begin{align}
\|\omega_{CB}\|_2^2&\le
a\Bigl(1+\frac{1}{d}-\Tr(\sigma_B^2)\Bigr)+b\,\Tr(\sigma_B^2)\nonumber\\
&=a\Bigl(1+\frac{1}{d}\Bigr)+\Bigl(b-a\Bigr)\Tr(\sigma_B^2)\nonumber\\
&\le a\Bigl(1+\frac{1}{d}\Bigr)+b-a= \frac{a}{d}+b,
\end{align}
where we used $\Tr(\sigma_B^2)\le 1$. Thus
\begin{equation}\label{eq:general-N-upper-ab}
\upsilon_{2}^{2}(\mathcal N) \le \max\left\{ a+\frac{b}{d},\ \frac{a}{d}+b \right\}.
\end{equation}
We next prove matching lower bounds.

For the first branch, choose the admissible output system $B=A'$ and the identity channel
\begin{equation}
\Lambda=\id_{A'}.
\end{equation}
Then
\begin{equation}
\tau_{\mathcal N}:=(\mathcal N\otimes \id_{A'})(\Phi_d)
\end{equation}
is an admissible output, and
\begin{align}
\|\tau_{\mathcal N}\|_2^2&=\Tr\Bigl(
\Phi_d\,\bigl((\mathcal N^\dagger\mathcal N)\otimes \id_{A'}\bigr)(\Phi_d)
\Bigr)\nonumber\\
&=\Tr\Bigl( \Phi_d\,
\bigl((a\,\id_A+b\,\Tr[\cdot]\,I_d)\otimes \id_{A'}\bigr)(\Phi_d)
\Bigr)\nonumber\\
&= a\,\Tr(\Phi_d^2) + b\,\Tr\Bigl( \Phi_d\Bigl(I_d\otimes \Tr_A(\Phi_d)\Bigr) \Bigr)\nonumber\\
&= a+b\,\Tr\Bigl( \Phi_d\Bigl(I_d\otimes \frac{I_d}{d}\Bigr) \Bigr)\nonumber\\
&= a+\frac{b}{d},
\end{align}
because $\Phi_d$ is pure and $\Tr_A(\Phi_d)=I_d/d$. Therefore
\begin{equation}
\upsilon_{2}^{2}(\mathcal N) \ge a+\frac{b}{d}.
\end{equation}
For the second branch, choose any finite dimensional output system $B$ and any unit vector $|\psi\rangle\in B$, and let
\begin{equation}
\Lambda_\psi(X)=\Tr(X)\,|\psi\rangle\langle\psi|.
\end{equation}
Then the corresponding admissible state is
\begin{equation}
\sigma_{AB}=\frac{I_d}{d}\otimes |\psi\rangle\langle\psi|,
\end{equation}
and therefore
\begin{equation}
(\mathcal N\otimes \id_B)(\sigma_{AB})=\mathcal N\,\left(\frac{I_d}{d}\right)\otimes |\psi\rangle\langle\psi|.
\end{equation}
Hence
\begin{equation}
\left\| (\mathcal N\otimes \id_B)(\sigma_{AB}) \right\|_2^2 =\left\| \mathcal N\,\left(\frac{I_d}{d}\right) \right\|_2^2.
\end{equation}
Now
\begin{align}
\left\| \mathcal N\,\left(\frac{I_d}{d}\right)
\right\|_2^2&=\frac{1}{d^2}\Tr\bigl(\mathcal N(I_d)^2\bigr)\nonumber\\
&=\frac{1}{d^2}\Tr\Bigl(I_d\,(\mathcal N^\dagger\mathcal N)(I_d)\Bigr)\nonumber\\
&=\frac{1}{d^2}\Tr\Bigl(I_d\,(aI_d+bd\,I_d)\Bigr)\nonumber\\
&=\frac{1}{d^2}\,d\,(a+bd)=\frac{a}{d}+b.
\end{align}
Therefore
\begin{equation}
\upsilon_{2}^{2}(\mathcal N) \ge \frac{a}{d}+b.
\end{equation}
Combining this with Eq.~\eqref{eq:general-N-upper-ab} proves Eq.~\eqref{eq:N-onecopy-closed-ab}.
We now turn to tensor powers. Let $\sigma_{A^{(n)}B^{(n)}}$ be admissible. Then, exactly as in the quadratic form argument for $\Gamma_t$,
\begin{align}
\left\| (\mathcal N^{\otimes n}\otimes \id_{B^{(n)}})(\sigma) \right\|_2^2
&= \Tr\Bigl( \sigma\, \bigl(((\mathcal N^\dagger\mathcal N)^{\otimes n})\otimes \id_{B^{(n)}}\bigr)(\sigma) \Bigr)\nonumber\\
&=\Tr\Bigl(\sigma\,\bigl((a\,\id_A+b\,\Tr[\cdot]\,I_d)^{\otimes n}\otimes \id_{B^{(n)}}\bigr)(\sigma)\Bigr).
\label{eq:general-N-ncopy-qf-ab}
\end{align}
Since
\begin{equation}
a\,\id_A+b\,\Tr[\cdot]\,I_d = a\,\id_A+bd\,D,
\end{equation}
the subset expansion gives
\begin{equation}
(a\,\id_A+b\,\Tr[\cdot]\,I_d)^{\otimes n}=\sum_{S\subseteq[n]}a^{|S|}(bd)^{|S^c|} \bigl(\id_{A_S}\otimes D_{A_{S^c}}\bigr).
\end{equation}
By the same marginal purity identity as before, we obtain
\begin{equation}
\left\|(\mathcal N^{\otimes n}\otimes \id_{B^{(n)}})(\sigma)\right\|_2^2 =\sum_{S\subseteq[n]} a^{|S|}b^{|S^c|} \Tr\bigl(\sigma_{A_SB^{(n)}}^2\bigr).
\end{equation}
Applying the same swap polynomial argument as in the alternative proof of all-$n$ multiplicativity, with the one copy factor now equal to
\begin{equation}
a\,U_k+b\,I,
\end{equation}
yields
\begin{equation}
\left\| (\mathcal N^{\otimes n}\otimes \id_{B^{(n)}})(\sigma)
\right\|_2^2 \le \left( \frac{1+\frac{1}{d}}{2}\,|a+b| + \frac{1-\frac{1}{d}}{2}\,|b-a| \right)^n.
\end{equation}
Since $a,b\ge 0$, this simplifies to
\begin{equation}
\left( \max\left\{ a+\frac{b}{d},\ \frac{a}{d}+b \right\}\right)^n.
\end{equation}
Therefore
\begin{equation}
\bigl(\upsilon_{2}(\mathcal N^{\otimes n})\bigr)^2 \le \left( \max\left\{ a+\frac{b}{d},\ \frac{a}{d}+b \right\} \right)^n.
\end{equation}
For the reverse inequality, take the product of the corresponding one copy optimizers from the two branches above. By multiplicativity of the Hilbert-Schmidt norm, this yields
\begin{equation}
\bigl(\upsilon_{2}(\mathcal N^{\otimes n})\bigr)^2 \ge \left( \max\left\{ a+\frac{b}{d},\ \frac{a}{d}+b \right\} \right)^n.
\end{equation}
Therefore the proposition follows. 
\end{proof}
We now discuss some natural classes of CP maps to which Proposition~\ref{prop:general-p2-criterion} applies. 

The depolarizing channel is defined by
\begin{equation}
\Delta_p(X) = (1-p)X + \frac{p}{d}\, I_d \operatorname{Tr}(X),
\qquad 0 \le p \le \frac{d^2}{d^2-1}.
\end{equation}
We now present the resulting corollary.

\begin{corollary}
The depolarizing channel $\Delta_p(\cdot)$, the transpose-depolarizing
channel $\Gamma_t(\cdot)$, and their corresponding complementary channels
$\Delta_p^c(\cdot)$ and $\Gamma_t^c(\cdot)$, all satisfy the condition of
Proposition~\ref{prop:general-p2-criterion}. 
Namely, for each such map $\cN$, there exist constants
$a_{\cN},b_{\cN}\ge 0$ such that
\begin{equation}
\cN^\dagger\circ \cN
=
a_{\cN}\,\id_A
+
b_{\cN}\,\Tr[\cdot]\,I_d,
\end{equation}
where
\begin{equation}
\cN\in\{\Delta_p,\Gamma_t,\Delta_p^c,\Gamma_t^c\}.
\end{equation}
Hence, for every $n\in\mathbb N$,
\begin{align*}
\upsilon_{2}^{(n)}(\Delta_p)
&=
\bigl(\upsilon_{2}(\Delta_p)\bigr)^n,\\
\upsilon_{2}^{(n)}(\Gamma_t)
&=
\bigl(\upsilon_{2}(\Gamma_t)\bigr)^n,\\
\upsilon_{2}^{(n)}(\Delta_p^c)
&=
\bigl(\upsilon_{2}(\Delta_p^c)\bigr)^n,\\
\upsilon_{2}^{(n)}(\Gamma_t^c)
&=
\bigl(\upsilon_{2}(\Gamma_t^c)\bigr)^n.
\end{align*}
\end{corollary}

\begin{proof}
Lemma~\ref{lem:Gamma-square} shows that the transpose-depolarizing
channel $\Gamma_t(\cdot)$ satisfies the condition of
Proposition~\ref{prop:general-p2-criterion}. An analogous argument shows
that the depolarizing channel $\Delta_p(\cdot)$ satisfies the same condition.

Moreover, Lemma~\ref{lemma: Gamma_t^c and Delta_p^c property} in the Appendix
shows that the complementary channels $\Gamma_t^c(\cdot)$ and
$\Delta_p^c(\cdot)$ also satisfy this condition. Hence, the claimed
multiplicativity follows from Proposition~\ref{prop:general-p2-criterion}.
\end{proof}

\begin{proposition}[Complementary invariance of $\upsilon_p$]\label{prop:upsilon-complementary-invariance}
Let
\[
\Omega:\mathcal L(B')\to \mathcal L(A)
\]
be a completely positive map. Let
\[
W:B'\to A\otimes K
\]
be a Stinespring operator for $\Omega$, i.e.
\begin{equation}\label{eq:Omega-Stinespring-operator}
\Omega(X)=\Tr_K(WXW^\dagger),
\qquad X\in\mathcal L(B').
\end{equation}
Define the corresponding complementary completely positive map (not necessarily trace-preserving)
\begin{equation}\label{eq:Omega-complementary-cp}
\Omega^c(X):=\Tr_A(WXW^\dagger),
\qquad X\in\mathcal L(B').
\end{equation}
Then, for every $1\le p\le \infty$,
\begin{equation}\label{eq:upsilon-complementary-invariance}
\upsilon_p(\Omega)=\upsilon_p(\Omega^c).
\end{equation}
\end{proposition}

\begin{proof}
Fix a CPTP map
\begin{equation}
\Lambda:\mathcal L(B)\to \mathcal L(E),
\end{equation}
and choose a Stinespring isometry
\begin{equation}
V:B\to E\otimes F
\end{equation}
for $\Lambda$, so that
\begin{equation}\label{eq:Lambda-Stinespring-isometry}
\Lambda(Y)=\Tr_F(VYV^\dagger),
\qquad Y\in\mathcal L(B).
\end{equation}
Let
\begin{equation}\label{eq:Lambda-complementary-channel}
\Lambda^c(Y):=\Tr_E(VYV^\dagger),
\qquad Y\in\mathcal L(B),
\end{equation}
denote the complementary channel of $\Lambda$.

Let $\Phi^{B'B}=|\Phi\rangle\langle\Phi|^{B'B}$ be the normalized maximally entangled state used in the definition of $\upsilon_p$, and define
\begin{equation}
|\Psi\rangle_{AKEF}:=(W\otimes V)|\Phi\rangle_{B'B}.
\end{equation}
Since $\Omega$ is only assumed to be completely positive, the operator $W$ need not be isometric; accordingly, $|\Psi\rangle_{AKEF}$ need not be normalized. This does not affect the argument below, since we only use that the two reduced operators come from the same rank one positive operator $|\Psi\rangle\langle\Psi|$.

Now using \eqref{eq:Omega-Stinespring-operator} and \eqref{eq:Lambda-Stinespring-isometry}, we obtain
\begin{align}
(\Omega\otimes \id_E)\bigl((\id_{B'}\otimes \Lambda)(\Phi^{B'B})\bigr)
&=\Tr_{KF}\bigl(|\Psi\rangle\langle\Psi|_{AKEF}\bigr).
\label{eq:first-reduction-cp}
\end{align}
Similarly, using \eqref{eq:Omega-complementary-cp} and \eqref{eq:Lambda-complementary-channel}, give us
\begin{align}
(\Omega^c\otimes \id_F)\bigl((\id_{B'}\otimes \Lambda^c)(\Phi^{B'B})\bigr)
&=\Tr_{AE}\bigl(|\Psi\rangle\langle\Psi|_{AKEF}\bigr).
\label{eq:second-reduction-cp}
\end{align}
The two operators in \eqref{eq:first-reduction-cp} and \eqref{eq:second-reduction-cp} are reduced operators of the same rank one positive operator $|\Psi\rangle\langle\Psi|_{AKEF}$. Hence they have the same nonzero spectrum, and therefore the same Schatten $p$-norm:
\begin{equation}\label{eq:equal-p-norms-cp}
\left\| (\Omega\otimes \id_E)\bigl((\id_{B'}\otimes \Lambda)(\Phi^{B'B})\bigr) \right\|_p=\left\| \Omega^c\otimes \id_F)\bigl((\id_{B'}\otimes \Lambda^c)(\Phi^{B'B})\bigr)
\right\|_p.
\end{equation}
Since every admissible CPTP map $\Lambda$ for $\upsilon_p(\Omega)$ produces an admissible CPTP map
$\Lambda^c$ for $\upsilon_p(\Omega^c)$, taking the supremum over all $\Lambda$ yields
\begin{equation}\label{eq:first-inequality-cp}
\upsilon_p(\Omega)\le \upsilon_p(\Omega^c).
\end{equation}
To obtain the reverse inequality, interchange the roles of $\Omega$ and $\Omega^c$. The same operator $W$
is also a Stinespring operator for $\Omega^c$, viewed as a completely positive map into $\mathcal L(K)$, and
the corresponding complementary completely positive map is precisely $\Omega$, obtained by tracing over $K$
instead of $A$. Repeating the argument above with the roles of $A$ and $K$ interchanged gives
\begin{equation}\label{eq:second-inequality-cp}
\upsilon_p(\Omega^c)\le \upsilon_p(\Omega).
\end{equation}
Combining \eqref{eq:first-inequality-cp} and \eqref{eq:second-inequality-cp} proves
\eqref{eq:upsilon-complementary-invariance}.
\end{proof}

\begin{corollary}[Transfer of multiplicativity to complementary CP maps]\label{cor:upsilon-complementary-multiplicativity}
For every $n\in\mathbb N$ and every $1\le p\le \infty$,
\begin{equation}\label{eq:upsilon-tensor-complementary}
\upsilon_p(\Omega^{\otimes n})=\upsilon_p((\Omega^c)^{\otimes n}).
\end{equation}
In particular, if
\begin{equation}\label{eq:Omega-multiplicative-assumption}
\upsilon_2(\Omega^{\otimes n})=\bigl(\upsilon_2(\Omega)\bigr)^n
\qquad \text{for all }n\in\mathbb N,
\end{equation}
then
\begin{equation}\label{eq:Omega-complement-multiplicative}
\upsilon_2((\Omega^c)^{\otimes n})=\bigl(\upsilon_2(\Omega^c)\bigr)^n
\qquad \text{for all }n\in\mathbb N.
\end{equation}
\end{corollary}

\begin{proof}
If $W:B'\to A\otimes K$ is a Stinespring operator for $\Omega$, then
\begin{equation}
W^{\otimes n}:(B')^{\otimes n}\to A^{\otimes n}\otimes K^{\otimes n}
\end{equation}
is a Stinespring operator for $\Omega^{\otimes n}$, and the corresponding complementary completely positive map is $(\Omega^c)^{\otimes n}$.
Therefore \eqref{eq:upsilon-tensor-complementary} follows from Proposition~\ref{prop:upsilon-complementary-invariance}.

Now assume \eqref{eq:Omega-multiplicative-assumption} applies. Then
\begin{equation}
\upsilon_2((\Omega^c)^{\otimes n})=\upsilon_2(\Omega^{\otimes n})
=\bigl(\upsilon_2(\Omega)\bigr)^n=\bigl(\upsilon_2(\Omega^c)\bigr)^n,
\end{equation}
where the first and third equalities use \eqref{eq:upsilon-tensor-complementary} and
Proposition~\ref{prop:upsilon-complementary-invariance}. This proves
\eqref{eq:Omega-complement-multiplicative}.
\end{proof}

\section{Discussion and Open Problems}\label{sec:summary}

In this paper, we show that the R\'{e}nyi entanglement of purification is
equivalently related to the constrained maximal output Schatten $p$-norm of
the CP map associated with the underlying bipartite state.
In particular, we study the constrained maximal output Schatten $2$-norm of a
CP map $\Omega$, defined by
\begin{align}\label{eq: upsilon_p Choi 2}
\upsilon_2(\Omega):=\sup_{\sigma^{B'E}: \,
\Tr_E(\sigma^{B'E})=I_{B'}/d_{B'}} \;
\left\|
(\Omega \otimes \id_E)\sigma^{B'E}
\right\|_{2}.
\end{align}
We first specialize to the case where $\Omega$ is given by the transpose-depolarizing channel:
\begin{equation}
    \Gamma_t(X)=tX^T + (1-t)\operatorname{Tr}(X)\frac{I_d}{d}, \qquad -\frac{1}{d-1} \le t\le \frac{1}{d+1}
\end{equation}
We obtained an exact evolution of the one copy problem at $p=2$. More precisely 
\begin{equation}
    \upsilon_2(\Gamma_t) = \sqrt{\max\left\{\frac{1}{d},\ t^2 +\frac{1-t^2}{d^2}\right\}}.
\end{equation}
These two branches are attained, respectively, at $\sigma^{B'E}= I_{B'}/d_{B'} \otimes \proj{\psi}^{E}$
for a pure state $\proj{\psi}^{E}$, and at a maximally entangled state on
$B'\otimes E$.
Further, we considered the tensor-power quantity $\upsilon_2(\Gamma_t^{\otimes n})$ and proved that it is multiplicative for every $n \in \mathbb{N}$ 
\begin{equation}
    \upsilon_2(\Gamma_t^{\otimes n})\,=\,\bigl(\upsilon_2(\Gamma_t)\bigr)^n.
\end{equation}
That is even when arbitrary joint CPTP post-processing is allowed across $n$-copies, entangled strategies do not improve the optimal output Schatten $2$-norm beyond product strategies. 
For $d\ge 3$, this expression further simplifies to
$\upsilon_2(\Gamma_t)=d^{-1/2}$
throughout the entire completely positive range of $t$.

The proof combines several complementary ideas. First, the $p=2$ objective admits a quadratic form representation, which reduces the problem to a weighted combination of global and marginal purities. Second, the feasible set yields the exact one copy upper bound, which matches the explicit lower bounds. Third, in the tensor power setting, the objective can be rewritten as a swap polynomial on a doubled space. The fixed marginal condition forces a product structure for the eigenvalue distribution, leading to the all-$n$ upper bound, which is matched by explicit product optimizers.

We then establish a general multiplicativity criterion. More precisely, for any
CP map $\cN$ satisfying
\begin{equation}
\cN^\dagger\circ \cN
=
a\,\id_A+b\,\Tr[\cdot]\,I_d
\end{equation}
for some $a,b\ge 0$, we compute $\upsilon_2(\cN)$ and prove that it is multiplicative under tensor powers. 

We further show that this criterion applies to the complementary channel $\Gamma_t^c(\cdot)$ of the transpose-depolarizing channel, as well as to the depolarizing channel $\Delta_p(\cdot)$ and its complementary channel $\Delta_p^c(\cdot)$. Consequently, $\upsilon_2$ is multiplicative for each of these channels. This provides a unified $p=2$ mechanism for multiplicativity across several natural channel families and their complementary channels.

\medskip

A natural next question is whether a similar multiplicativity statement holds for the Hilbert-Schmidt adjoint of the complementary channel of $\Gamma_t$, namely $(\Gamma_t^c)^{\dagger}$.
The multiplicativity problem for $(\Gamma_t^c)^{\dagger}$ appears to be a
distinct question that requires techniques beyond those developed in this work.
This map is of particular interest because it arises naturally in the study of the entanglement of purification for Werner states, defined as follows
\begin{equation}
W_{AC}(t)
:=(\operatorname{id}_A\otimes \Gamma_t)(\Phi_d^{AA'})=
\frac{1-t}{d^2}I_{AC}
+\frac{t}{d}\Pi_{AC}.
\end{equation}
where $\Pi_{AC}$ is the swap operator on $A\otimes C$.  
In particular, it is known that the entanglement of purification is additive
for $t=0,\frac{1}{1-d},\frac{1}{1+d}$ \cite{Christandl-Winter2005,Terhal_2002}. A systematic treatment of this problem is left for future work. In particular, understanding the multiplicativity properties of $(\Gamma_t^c)^\dagger$ may provide a direct route to corresponding additivity questions for Rényi entanglement of purification of Werner states.
Moreover, in dimension $d=2$, numerical evidence suggests that, up to possible
numerical error, the entanglement of purification is not additive \cite{Chen-Winter2012}.

We show in Lemma~\ref{lemma: map for Werner} that for Werner states the corresponding CP map appearing in Eq.~\ref{eq: upsilon_p Choi 2} is given by $(d\,\Gamma_t^c)^{\dagger}$. Consequently, the additivity problem for the R\'{e}nyi entanglement of purification  for Werner states is equivalent to the multiplicativity problem for the following quantity:
\begin{equation}
\upsilon_p\left(d\,(\Gamma_t^c)^{\dagger}\right)
:= \sup_{\Lambda:\mathcal L(B_1\otimes B_2)\to\mathcal L(E)\ \mathrm{CPTP}}
\left\| d\,\bigl((\Gamma_t^c)^{\dagger}\otimes \id_E\bigr)
\bigl((\id_{B'_1B'_2}\otimes \Lambda)(\Phi_{d}^{B'_1 B_1}\otimes \Phi_{d}^{B'_2 B_2})\bigr) \right\|_p .
\end{equation}
where $(\Gamma_t^c)^{\dagger}:\cL(B'_1 \otimes B'_2)\to \cL(A)$ is CP,
$\Lambda:\cL(B_1\otimes B_2)\to \cL(E)$ is CPTP, and
$\Phi_{d}^{B'_1 B_1}$ and $\Phi_{d}^{B'_2 B_2}$ are normalized maximally
entangled states of dimension $d$.

Finally, we obtain the following explicit form of
$(\Gamma_t^c)^{\dagger}$ in Lemma~\ref{lemma: map for Werner}:
\begin{align}
  (\Gamma_t^c)^{\dagger}(Y) =\Tr_{B'_2}(S_t \; Y \;S_t)  
\end{align}
where
\begin{equation}
S_t =
(a^++a^-)I_{B'_1B'_2}+
(a^+-a^-)\Pi_{B'_1B'_2}.
\end{equation}
Here $\Pi_{B'_1B'_2}$ is the flip operator on $B'_1\otimes B'_2$ and
\begin{equation}
a^+
=\sqrt{\frac{1+(d-1)t}{4d}},
\qquad
a^-
=\sqrt{\frac{1-(d+1)t}{4d}}.
\end{equation}

\bigskip
\noindent \textbf{Acknowledgments.}
SF acknowledges support from the University of Waterloo, Natural Sciences and Engineering Research Council of Canada, and the Government of Canada through the Department of Innovation, Science and Economic Development and by the Province of Ontario through the Ministry of Colleges and Universities at Perimeter Institute.
ZBK acknowledges support from the Ada Lovelace Postdoctoral Fellowship at the Perimeter Institute for Theoretical Physics.

\appendix

\section{State-Map Correspondence}
\begin{lemma}[state to map representation]\label{lemma: state-CP map}
Let $A$ and $B$ be finite dimension quantum systems, and let
$B'\cong B$. For any bipartite quantum state $\rho^{AB}$, there exists a
completely positive map
\begin{equation}
\Omega:\mathcal L(B')\to \mathcal L(A)
\end{equation}
such that
\begin{equation}
\rho^{AB}
=
(\Omega\otimes \id_{B})(\Phi^{B'B}),
\end{equation}
where
\begin{equation}
\ket{\Phi}^{B'B}
=
\frac{1}{\sqrt{d_B}}
\sum_{i=1}^{d_B}\ket{i}^{B'}\ket{i}^{B},
\qquad
\Phi^{B'B}=\proj{\Phi}^{B'B},
\end{equation}
and $d_B=\dim B$. The map $\Omega$ is not necessarily trace-preserving.
\end{lemma}
\begin{proof}
Since $\rho^{AB}\geq 0$, we may write
\begin{equation}
\rho^{AB}
=
\sum_\alpha \proj{\psi_\alpha}^{AB},
\end{equation}
where the vectors $\ket{\psi_\alpha}^{AB}$ are not necessarily normalized.
For instance, one may absorb the eigenvalues of $\rho^{AB}$ into the
corresponding eigenvectors.

Fix orthonormal bases $\{\ket{a}^{A}\}_a$ of $A$ and
$\{\ket{i}^{B}\}_i$ of $B$. For each $\alpha$, write
\begin{equation}
\ket{\psi_\alpha}^{AB}
=
\sum_{a,i} c^{(\alpha)}_{a i}
\ket{a}^{A}\ket{i}^{B},
\end{equation}
And define an operator $K_\alpha:B'\to A$ by
\begin{equation}
K_\alpha
=
\sqrt{d_B}
\sum_{a,i}
c^{(\alpha)}_{a i}
\ket{a}^{A}\bra{i}^{B'}.
\end{equation}
Then, by the definition of the maximally entangled state,
\begin{equation}
(K_\alpha\otimes I_B)\ket{\Phi}^{B'B}
=
\ket{\psi_\alpha}^{AB}.
\end{equation}
Now define
\begin{equation}
\Omega(X)
=
\sum_\alpha K_\alpha X K_\alpha^\dagger,
\qquad
X\in \mathcal L(B').
\end{equation}
This map is completely positive, since it is written in Kraus form. Hence
\begin{equation}
\begin{aligned}
(\Omega\otimes \id_B)(\Phi^{B'B})
&=
\sum_\alpha
(K_\alpha\otimes I_B)
\Phi^{B'B}
(K_\alpha^\dagger\otimes I_B) \\
&=
\sum_\alpha \proj{\psi_\alpha}^{AB} \\
&=
\rho^{AB}.
\end{aligned}
\end{equation}
Thus every bipartite state $\rho^{AB}$ can be represented as the image of
one half of a maximally entangled state under a completely positive map.
This is precisely the Choi-Jamio{\l}kowski representation of
$\rho^{AB}$ in the normalization convention determined by $\Phi^{B'B}$.

\end{proof}

\section{Channels Satisfying the General Criterion}
\begin{lemma}[Adjoint square identities for $\Gamma_t^c$ and $\Delta_p^c$]\label{lemma: Gamma_t^c and Delta_p^c property}
Let $\Gamma_t^c : \cL(\mathbb{C}^d) \to \cL(\mathbb{C}^d\otimes \mathbb{C}^d)$  be the complementary channel of the transpose-depolarizing
channel in the Datta-Fukuda-Holevo representation \cite{DattaFukudaHolevo2006,Holevo2007ComplementaryChannels}
\begin{equation}
\Gamma_t^c(X)=S_t(X\otimes I)S_t^\dagger, \nonumber
\end{equation}
where
\begin{equation}
S_t=(a^++a^-)I_{12}+(a^+-a^-)\Pi_{12}, \nonumber
\end{equation}
where $I_{12}$ and $\Pi_{12}$ are respectively the identity and  flip operator on $\mathbb{C}^d\otimes \mathbb{C}^d$,
and 
\begin{equation}
a^+
=
\sqrt{\frac{1+(d-1)t}{4d}},
\qquad
a^-
=
\sqrt{\frac{1-(d+1)t}{4d}}.  \nonumber
\end{equation}
Then, for every $X\in \cL(\mathbb{C}^d)$,
\begin{equation}
\bigl((\Gamma_t^c)^\dagger\circ \Gamma_t^c\bigr)(X)
=
\frac{1-t^2}{d}X+t^2\Tr[X]I_d. \nonumber
\end{equation}
Similarly, let $\Delta_p^c : \cL(\mathbb{C}^d) \to \cL(\mathbb{C}^d\otimes \mathbb{C}^d)$  be the complementary channel of the depolarizing
channel in the Datta-Fukuda-Holevo representation
\begin{equation}
\Delta_p^c(X)=P_p(X\otimes I_d)P_p^\dagger,
\end{equation}
where
\begin{equation}
P_p
=
\sqrt{\frac{p}{d}}\, I_{12}
+
\sqrt d
\left(
-\frac{\sqrt p}{d}
+
\sqrt{1-p\left(\frac{d^2-1}{d^2}\right)}
\right)
\proj{\Phi_d},
\end{equation}
and
\begin{equation}
\ket{\Phi_d}
=
\frac{1}{\sqrt d}\sum_{i=1}^d \ket{i}\otimes \ket{i}.
\end{equation}
Then, for every $X\in \cL(\mathbb C^d)$,
\begin{equation}
\bigl((\Delta_p^c)^\dagger\circ \Delta_p^c\bigr)(X)
=
\frac{p(2-p)}{d}X
+
(1-p)^2\Tr[X]I_d.
\end{equation}
\end{lemma}
\begin{proof}
We compute the Hilbert-Schmidt adjoint
\begin{equation}
(\Gamma_t^c)^\dagger:
\mathcal{L}(\mathbb{C}^d\otimes \mathbb{C}^d)
\to
\mathcal{L}(\mathbb{C}^d).
\end{equation}
The Hilbert-Schmidt adjoint is defined by
\begin{equation}
\Tr\,\left[
Y^\dagger \Gamma_t^c(X)
\right]
=
\Tr\,\left[
\bigl((\Gamma_t^c)^\dagger(Y)\bigr)^\dagger X
\right]
\end{equation}
for all $X\in\mathcal{L}(\mathbb{C}^d)$ and
$Y\in\mathcal{L}(\mathbb{C}^d\otimes \mathbb{C}^d)$.

By the definition of $\Gamma_t^c$, we obtain
\begin{align}
\Tr\,\left[
Y^\dagger \Gamma_t^c(X)
\right]
&=
\Tr\,\left[
Y^\dagger S(X\otimes I_d)S^\dagger
\right] \nonumber\\
&=
\Tr\,\left[
S^\dagger Y^\dagger S(X\otimes I_d)
\right] \nonumber\\
&=
\Tr\,\left[
(S^\dagger YS)^\dagger (X\otimes I_d)
\right].
\end{align}
Now we use the defining property of the partial trace:
for every $Z\in\mathcal{L}(\mathbb{C}^d\otimes \mathbb{C}^d)$,
\begin{equation}
\Tr\,\left[
Z^\dagger (X\otimes I_d)
\right]
=
\Tr\,\left[
(\Tr_2 Z)^\dagger X
\right].
\end{equation}
Applying this with $Z=S^\dagger YS$ gives
\begin{align}
\Tr\,\left[
Y^\dagger \Gamma_t^c(X)
\right]
&=
\Tr\,\left[
\bigl(\Tr_2[S^\dagger YS]\bigr)^\dagger X
\right].
\end{align}
Comparing with the definition of the Hilbert-Schmidt adjoint, we obtain
\begin{equation}
(\Gamma_t^c)^\dagger(Y)
=
\Tr_2[S^\dagger YS].
\end{equation}
In the Datta-Fukuda-Holevo representation $S:=S^\dagger$, so
\begin{equation}
(\Gamma_t^c)^\dagger(Y)
=
\Tr_2[SYS].
\end{equation}
Therefore
\begin{align}\label{eq:combin-Gammas}
\bigl((\Gamma_t^c)^\dagger\circ \Gamma_t^c\bigr)(X)
&=
(\Gamma_t^c)^\dagger\,\left(S(X\otimes I_d)S\right) \nonumber\\
&=
\Tr_2\,\left[
S^2(X\otimes I_d)S^2
\right].
\end{align}
Now we set
\begin{equation}
\alpha:=a^++a^-,
\qquad
\beta:=a^+-a^-,
\end{equation}
and 
\begin{equation}
S=\alpha I_{12}+\beta \Pi.
\end{equation}
Since $\Pi^2=I_{12}$, we obtain
\begin{equation}
S^2=
(\alpha I_{12}+\beta \Pi)^2
=
(\alpha^2+\beta^2)I_{12}+2\alpha\beta \Pi.
\end{equation}
By using
\begin{equation}
(a^+)^2=\frac{1+(d-1)t}{4d},
\qquad
(a^-)^2=\frac{1-(d+1)t}{4d},
\end{equation}
we also obtain
\begin{align}
&\alpha^2+\beta^2=2\bigl((a^+)^2+(a^-)^2\bigr) =\frac{1-t}{d},\\
&2\alpha\beta= 2\bigl((a^+)^2-(a^-)^2\bigr)=t.
\end{align}
Hence 
\begin{equation}
S^2=\frac{1-t}{d}I_d+t\,\Pi.
\end{equation}
Thus we can write Eq.~\eqref{eq:combin-Gammas} as
\begin{align}
\bigl((\Gamma_t^c)^\dagger\circ \Gamma_t^c\bigr)(X)
&=
\Tr_2\,\left[
\left(\frac{1-t}{d}I_d+t\,\Pi\right)
(X\otimes I_d)
\left(\frac{1-t}{d}I_d+t\,\Pi\right)
\right].
\end{align}
Expanding gives
\begin{align}\label{eq:comb-Gammas2}
\bigl((\Gamma_t^c)^\dagger\circ \Gamma_t^c\bigr)(X)
&=
\frac{(1-t)^2}{d^2}\Tr_2[X\otimes I_d] \nonumber\\
&\quad
+
\frac{t(1-t)}{d}
\Tr_2[(X\otimes I_d)\,\Pi] \nonumber\\
&\quad
+
\frac{t(1-t)}{d}
\Tr_2[\,\Pi(X\otimes I_d)] \nonumber\\
&\quad
+
t^2\Tr_2[\,\Pi(X\otimes I_d)\,\Pi].
\end{align}
Now we use the standard identities $\Tr_2[X\otimes I_d]=dX$, and 
\begin{align}
&\Tr_2[(X\otimes I_d)\,\Pi]=X,
\qquad
\Tr_2[\,\Pi(X\otimes I_d)]=X,\\
&\Pi(X\otimes I_d)\,\Pi=I_d\otimes X,
\qquad
\Tr_2[I_d\otimes X]=\Tr[X]I_d,
\end{align}
and substituting them into Eq.~\eqref{eq:comb-Gammas2} gives
\begin{align}
\bigl((\Gamma_t^c)^\dagger\circ \Gamma_t^c\bigr)(X)
&=
\frac{(1-t)^2}{d^2}dX
+
\frac{t(1-t)}{d}X
+
\frac{t(1-t)}{d}X
+
t^2\Tr[X]I_d \nonumber\\
&=
\left(
\frac{(1-t)^2}{d}
+
\frac{2t(1-t)}{d}
\right)X
+
t^2\Tr[X]I_d \nonumber \\
&=
\frac{1-t^2}{d}X+t^2\Tr[X]I_d.
\end{align}
Therefore the claim follows
\begin{equation}
(\Gamma_t^c)^\dagger\circ \Gamma_t^c
=
\frac{1-t^2}{d}\operatorname{id}
+
t^2\Tr[\cdot]I_d.
\end{equation}
Now, we prove the second statement. We  compute the Hilbert-Schmidt adjoint of $\Delta_p^c$.
For $X\in \cL(\mathbb C^d)$ and
$Y\in \cL(\mathbb C^d\otimes \mathbb C^d)$, we have
\begin{align}
\Tr\left[Y^\dagger \Delta_p^c(X)\right]
&=
\Tr\left[
Y^\dagger P_p(X\otimes I_d)P_p^\dagger
\right] \nonumber\\
&=
\Tr\left[
P_p^\dagger Y^\dagger P_p(X\otimes I_d)
\right] \nonumber\\
&=
\Tr\left[
(P_p^\dagger YP_p)^\dagger (X\otimes I_d)
\right].
\end{align}
Using the defining property of the partial trace,
\begin{equation}
\Tr\left[
Z^\dagger (X\otimes I_d)
\right]
=
\Tr\left[
(\Tr_2 Z)^\dagger X
\right],
\end{equation}
with $Z=P_p^\dagger YP_p$, we obtain
\begin{equation}
(\Delta_p^c)^\dagger(Y)
=
\Tr_2[P_p^\dagger YP_p].
\end{equation}
Since $P_p=P_p^\dagger$, this becomes
\begin{equation}
(\Delta_p^c)^\dagger(Y)
=
\Tr_2[P_pYP_p].
\end{equation}
Therefore,
\begin{align}
\bigl((\Delta_p^c)^\dagger\circ \Delta_p^c\bigr)(X)
&=
\Tr_2\left[
P_p^2(X\otimes I_d)P_p^2
\right].
\end{align}
It remains to compute $P_p^2$. Set
\begin{equation}
Q:=\proj{\Phi_d},
\qquad
\gamma
:=
\sqrt{1-p\left(\frac{d^2-1}{d^2}\right)}.
\end{equation}
Then
\begin{equation}
P_p
=
\sqrt{\frac pd}\,I_{12}
+
\sqrt d\left(
-\frac{\sqrt p}{d}+\gamma
\right)Q.
\end{equation}
Since $Q^2=Q$, the eigenvalue of $P_p$ on $Q^\perp$ is
\begin{equation}
\sqrt{\frac pd},
\end{equation}
while the eigenvalue on the span of $\ket{\Phi_d}$ is
\begin{equation}
\sqrt{\frac pd}
+
\sqrt d\left(
-\frac{\sqrt p}{d}+\gamma
\right)
=
\sqrt d\,\gamma.
\end{equation}
Hence
\begin{align}
P_p^2
&=
\frac pd I_{12}
+
\left(
d\gamma^2-\frac pd
\right)Q \nonumber\\
&=
\frac pd I_{12}
+
\left(
d\left(1-p\frac{d^2-1}{d^2}\right)-\frac pd
\right)Q \nonumber\\
&=
\frac pd I_{12}
+
d(1-p)Q.
\end{align}
Thus, if $A:=\frac pd$ and $B:=d\,(1-p)$, then
\begin{equation}
P_p^2=A I_{12}+BQ.
\end{equation}
Therefore,
\begin{align}
\bigl((\Delta_p^c)^\dagger\circ \Delta_p^c\bigr)(X)
&=
\Tr_2\left[
(AI_{12}+BQ)(X\otimes I_d)(AI_{12}+BQ)
\right] \nonumber\\
&=
A^2\Tr_2[X\otimes I_d]
+
AB\Tr_2[(X\otimes I_d)Q] \nonumber\\
&\quad
+
AB\Tr_2[Q(X\otimes I_d)]
+
B^2\Tr_2[Q(X\otimes I_d)Q].
\end{align}
We use the standard identities
\begin{align}
\Tr_2[X\otimes I_d]&=dX,\\
\Tr_2[(X\otimes I_d)Q]&=\frac{1}{d}X,\\
\Tr_2[Q(X\otimes I_d)]&=\frac{1}{d}X,\\
\Tr_2[Q(X\otimes I_d)Q]&=\frac{\Tr[X]}{d^2}I_d.
\end{align}
Substituting these identities gives
\begin{align}
\bigl((\Delta_p^c)^\dagger\circ \Delta_p^c\bigr)(X)
&=
A^2 dX
+
\frac{AB}{d}X
+
\frac{AB}{d}X
+
\frac{B^2}{d^2}\Tr[X]I_d \nonumber\\
&=
\left(
A^2d+\frac{2AB}{d}
\right)X
+
\frac{B^2}{d^2}\Tr[X]I_d.
\end{align}
Since $A=\frac{p}{d}$ and $B=d\,(1-p)$, we get
\begin{equation}
A^2d=\frac{p^2}{d}, \qquad
\frac{2AB}{d} = \frac{2p(1-p)}{d},
\qquad \frac{B^2}{d^2}=(1-p)^2.
\end{equation}
Hence
\begin{align}
\bigl((\Delta_p^c)^\dagger\circ \Delta_p^c\bigr)(X)
&=
\left(
\frac{p^2}{d}
+
\frac{2p(1-p)}{d}
\right)X
+
(1-p)^2\Tr[X]I_d \nonumber\\
&=
\frac{p(2-p)}{d}X
+
(1-p)^2\Tr[X]I_d.
\end{align}
Therefore,
\begin{equation}
(\Delta_p^c)^\dagger\circ \Delta_p^c
=
\frac{p(2-p)}{d}\operatorname{id}
+
(1-p)^2\Tr[\cdot]I_d.
\end{equation}
\end{proof}

\section{Transpose trick for channels and the Werner State Map} \label{sec:transpose-trick}

\begin{lemma}[Transpose trick for channels]\label{lemma: Transpose channels}
Let $A'\simeq A$ and $B'\simeq B$, and fix bases for all four systems. Let $\cN:\cL(A')\to \cL(B)$ be a channel, where
$\dim A'=d_A$ and $\dim B=d_B$. Then
\begin{align}
    (\id_A \otimes \cN) (\proj{\Phi_{d_A}}_{AA'})= \frac{d_B}{d_A}\; \left( (T_A\circ \cN^{\dagger} \circ T_{B'}) \otimes \id_{B}  \right)  (\proj{\Phi_{d_B}}_{B'B}),
\end{align}
where $T_A$ and $T_{B'}$ denote transposition with
respect to  the fixed bases of system $A$ and  $B'$, respectively.
The map  $\cN^\dagger:\cL(B')\to \cL(A)$ is the Hilbert-Schmidt adjoint
of $\cN$, and $\Phi_{d_A}^{AA'}$, and $\Phi_{d_B}^{B'B}$ are  normalized
maximally entangled states in the fix bases.
\end{lemma}

\begin{proof}
Using the transpose trick for a rectangular operator, let 
$M:A'\to B$, where $\dim A=d_A$ and $\dim B=d_B$. Let
\begin{align}
    |\Phi_{d_A}\rangle_{AA'} =\frac{1}{\sqrt{d_A}}\sum_{i=1}^{d_A}|e_i\rangle_A|e_i\rangle_{A'} \quad \text{and} \quad
    |\Phi_{d_B}\rangle_{B'B}=
\frac{1}{\sqrt{d_B}}\sum_{j=1}^{d_B}
|f_j\rangle_{B'}|f_j\rangle_{B},
\end{align}
where the states are expressed with respect to the fixed bases of the corresponding systems.
Then
\begin{equation}
(I_{A} \otimes M)|\Phi_{d_A}\rangle_{AA'}
=
\sqrt{\frac{d_B}{d_A}}\,
(M^T\otimes I_{B'} )|\Phi_{d_B}\rangle_{B'B} .
\end{equation}
where $M^T:{B' \to A}$, and  the transpose is with respect to the fixed bases.

Let $\cN:\cL(A')\to \cL(B)$ be a completely positive map with Kraus representation
\begin{equation}
\cN(\rho)=\sum_k M_k \rho M_k^{\dagger}.
\end{equation}
Then the Hilbert-Schmidt adjoint map
$\cN^{\dagger}:\cL(B')\to \cL(A)$ admits the Kraus representation
\begin{equation}
\cN^{\dagger}(\sigma)
=\sum_k M_k^{\dagger}\sigma M_k .
\end{equation}
By applying  the transpose trick for the Kraus operators, we obtain
\begin{align}
    (\id \otimes \cN)\proj{\Phi_{d_A}}_{AA'}&=\sum_k (I_{A} \otimes M_k) \proj{\Phi_{d_A}}_{AA'} (I_{A} \otimes M_k^{\dagger})\\
    &=\frac{d_B}{d_A}\sum_k (M_k^T \otimes I_{B}) \proj{\Phi_{d_B}}_{B'B} (M_k^{*} \otimes I_{B}).
\end{align}
Let $\{N_k=M_k^T\}_k$ be the Kraus operators of a CP map $\cM$. 
Using the fact that $\{M_k^{\dagger}\}_k$ are Kraus operators of the adjoint maps, it is straightforward to verify that 
\begin{align}
    \cM= \; T_A\circ \cN^{\dagger} \circ T_{B'}. 
\end{align}
Substituting this into the previous display gives the claimed identity.
\end{proof}

We now recall the definition of Werner states. Let $A\simeq C\simeq \mathbb{C}^d$, and fix computational bases on both
systems. Let
\begin{equation}
|\Phi\rangle_{AA'}
=
\frac{1}{\sqrt d}\sum_{i=1}^d |i\rangle_A\otimes |i\rangle_{A'}
\end{equation}
be the normalized maximally entangled vector, and let
$\Phi_{AA'}=\proj{\Phi}_{AA'}$. For
\begin{equation}
-\frac{1}{d-1}\le t\le \frac{1}{d+1},
\end{equation}
define the transpose-depolarizing channel 
$\Gamma_t: {\mathcal{L}(A')\to \mathcal{L}(C)}$ by
\begin{equation}
\Gamma_t(X)
=
t X^T
+
(1-t)\operatorname{Tr}(X)\frac{I_C}{d},
\end{equation}
where the transpose is taken with respect to the fixed computational basis.
The corresponding Werner state on $A\otimes C$ is defined as
\begin{equation}
W_{AC}(t)
:=
(\operatorname{id}_A\otimes \Gamma_t)(\Phi_{AA'}).
\end{equation}
Using
\begin{equation}
(\operatorname{id}_A\otimes T)(\Phi_{AA'})
=
\frac{1}{d}\Pi_{AC},
\end{equation}
where $\Pi_{AC}$ is the swap operator on $A\otimes C$, we obtain
\begin{equation}
W_{AC}(t)
=
\frac{1-t}{d^2}I_{AC}
+
\frac{t}{d}\Pi_{AC}.
\end{equation}
\begin{lemma}\label{lemma: map for Werner}
Let $A\simeq C$ have dimension $d_A$, and let $B'_1\simeq B_1\simeq B'_2\simeq B_2$ have  dimension $d_B$. For a Werner state $W_{AC}(t)$ there is a purification
$\ket{W(t)}_{ACB_1B_2}$ such that the reduced state on $AB_1B_2$ is obtained as
\begin{align}
W^c_{AB_1B_2}(t):=
\Tr_C\,\left[\proj{W(t)}_{ACB_1B_2}\right]
&= d_A\; \left( (\Gamma_t^c)^{\dagger} \otimes \id_{B_1B_2}  \right)  (\Phi_{d_B}^{B'_1 B_1} \otimes \Phi_{d_B}^{B'_2 B_2})
\end{align}
where $\Phi_{d_B}^{B'_1 B_1}$ and $\Phi_{d_B}^{B'_2 B_2}$  are 
normalized
maximally entangled state of dimension $d_B= d_A=d$.
Furthermore, 
$(\Gamma_t^c)^{\dagger}:\cL(B'_1 \otimes B'_2)\to \cL(A)$ is the adjoint of the complementary channel of
$\Gamma_t$. This is a CP map obtained as $(\Gamma_t^c)^{\dagger}(Y) =\Tr_{B'_2}(S_t \; Y \;S_t)$, where 
\begin{equation}
S_t
=
(a^++a^-)I_{B'_1B'_2}
+
(a^+-a^-)\Pi_{B'_1B'_2}.
\end{equation}
Here $\Pi_{B'_1B'_2}$ is the flip operator on $B'_1\otimes B'_2$, and
\begin{equation}
a^+
=
\sqrt{\frac{1+(d-1)t}{4d}},
\qquad
a^- = \sqrt{\frac{1-(d+1)t}{4d}}.
\end{equation}
\end{lemma}

\vspace{-0.7cm}

\begin{proof}
Let
\begin{equation}
W_{AC}(t)=
(\operatorname{id}_A\otimes \Gamma_t)(\Phi_d^{AA'})
\end{equation}
be the Werner state obtained from the transpose-depolarizing channel
$\Gamma_t:\cL(A')\to \cL(C)$. Let
\begin{equation}
V:A'\to  C\otimes B_1 \otimes B_2
\end{equation}
be a Stinespring isometry for $\Gamma_t$, so that
\begin{equation}
\Gamma_t(X)
=
\Tr_{B_1B_2}[VXV^\dagger].
\end{equation}
Define a purification of $W_{AC}(t)$ by
\begin{equation}
\ket{W(t)}_{ACB_1B_2}
:=
(\operatorname{id}_A\otimes V_{A'\to CB_1B_2})\ket{\Phi_{d_A}}_{AA'}.
\end{equation}
Indeed,
\begin{align}
\Tr_{B_1B_2}\,\left[\proj{W(t)}_{ACB_1B_2}\right]
&=(\operatorname{id}_A\otimes \Tr_{B_1B_2})
\left[
(\operatorname{id}_A\otimes V)
\Phi_{d_A}^{AA'}
(\operatorname{id}_A\otimes V^\dagger)
\right] \\
&=
(\operatorname{id}_A\otimes \Gamma_t)(\Phi_{d_A}^{AA'}) \\
&=
W_{AC}(t).
\end{align}
Hence $\ket{W(t)}_{ACB_1B_2}$ is a purification of $W_{AC}(t)$.
The reduced state on $AB_1B_2$ is obtained by tracing out $C$. Define
\begin{align}
W^c_{AB_1B_2}(t)
&:=\Tr_C\,\left[\proj{W(t)}_{ACB_1B_2}\right] \\
&=
(\operatorname{id}_A\otimes \Gamma_t^c)(\Phi_{d_A}^{AA'}),
\end{align}
where in the second line the complementary channel associated with the Stinespring isometry
$V$ is
\begin{equation}
\Gamma_t^c(X)
=
\Tr_C[VXV^\dagger].
\end{equation}
The complementary channel of the transpose-depolarizing channel $\Gamma_t^c: \cL(A')\to \cL(B_1 \otimes B_2)$ can be
chosen in the minimal representation of Datta-Fukuda-Holevo in \cite{DattaFukudaHolevo2006,Holevo2007ComplementaryChannels} 
\begin{equation}
\Gamma_t^c(\rho)
=
S_t(\rho\otimes I)S_t^\dagger,
\end{equation}
where
\begin{equation}
S_t
=
(a^++a^-)I_{B_1B_2}
+
(a^+-a^-)\Pi_{B_1B_2}.
\end{equation}
Here $\Pi_{B_1B_2}$ is the flip operator, 
and $d=d_A=d_B$ 
and
\begin{equation}
a^+
=
\sqrt{\frac{c^+}{2(d+1)}},
\qquad
a^-
=
\sqrt{\frac{c^-}{2(d-1)}},
\end{equation}
with
\begin{equation}
c^+
=
\frac{d^2-1}{2d}
\left(\frac{1}{d-1}+t\right),
\qquad
c^-
=
\frac{d^2-1}{2d}
\left(\frac{1}{d+1}-t\right).
\end{equation}
Equivalently,
\begin{equation}
a^+
=
\sqrt{\frac{1+(d-1)t}{4d}},
\qquad
a^-
=
\sqrt{\frac{1-(d+1)t}{4d}}.
\end{equation}
The Kraus operators of the map $\Gamma_t^c$ are $\{S_t(I_1 \otimes \ket{j})\}_{j=1}^d$ where $\{\ket{j}\}_{j=1}^d$ denotes the computational basis.
Hence, the Kraus operators of the adjoint map $(\Gamma_t^c)^{\dagger}$ are $\{(I_1 \otimes \bra{j})S_t\}_{j=1}^d$, which follows because $S_t^T=S_t^{\dagger}=S_t$.
Using the Kraus representation of the channel, it is straightforward to verify that
\begin{equation}
T_A\circ (\Gamma_t^c)^{\dagger}\circ T_{B'} = (\Gamma_t^c)^{\dagger}.
\end{equation}
Identifying
\begin{equation}
\Phi_{d_B^2}^{(B'_1B'_2)(B_1B_2)}=
\Phi_{d_B}^{B'_1B_1}\otimes \Phi_{d_B}^{B'_2B_2}
\end{equation}
with respect to the product computational basis, we may apply Lemma~\ref{lemma: Transpose channels} to the channel
\begin{equation}
\Gamma_t^c:\mathcal L(A')\to \mathcal L(B_1\otimes B_2),
\end{equation}
whose output dimension is $d_B^2=d_A^2$. Finally we obtain
\begin{align}
W^c_{AB_1B_2}(t)
&= d\; \left( (\Gamma_t^c)^{\dagger} \otimes \id_{B_1B_2}  \right)  ((\Phi_d^{B'_1 B_1} \otimes \Phi_d^{B'_2 B_2}))
\end{align}
where $(\Gamma_t^c)^{\dagger}(Y)= \sum_j (I_{B'_1} \otimes \bra{j})S_t \; Y \;S_t(I_{B'_1} \otimes \ket{j}) =\Tr_{B'_2}(S_t \; Y \;S_t)$.

\end{proof}

\bibliographystyle{unsrt}
\bibliography{max-multip=2}

\end{document}